\documentclass[a4paper,UKenglish,cleveref, autoref, thm-restate]{lipics-v2019}

\newif\iflongversion
\longversionfalse

\usepackage{tikz}
\usetikzlibrary{calc}
\usetikzlibrary{patterns}
\usepackage{csquotes}

\usepackage{mdframed}

\crefname{observation}{observation}{observations}
\Crefname{observation}{Observation}{Observations}

\title{The Complexity of Blocking All Solutions}

\author{Christoph Grüne}{Department of Computer Science, RWTH Aachen University, Germany}{gruene@algo.rwth-aachen.de}{https://orcid.org/0000-0002-7789-8870}{Funded by the German Research Foundation (DFG) – GRK 2236/2.}
\author{Lasse Wulf}{Section of Algorithms, Logic and Graphs, Technical University of Denmark, Kongens Lyngby, Denmark}{lawu@dtu.dk}{https://orcid.org/0000-0001-7139-4092}{Funded by the Carlsberg Foundation CF21-0302 ``Graph Algorithms with Geometric Applications''.}

\authorrunning{C. Grüne and L. Wulf} 

\Copyright{Christoph Grüne and Lasse Wulf} 



\category{} 

\relatedversion{} 

\supplement{}

\nolinenumbers 

\hideLIPIcs  

\EventEditors{John Q. Open and Joan R. Access}
\EventNoEds{2}
\EventLongTitle{42nd Conference on Very Important Topics (CVIT 2016)}
\EventShortTitle{CVIT 2016}
\EventAcronym{CVIT}
\EventYear{2016}
\EventDate{December 24--27, 2016}
\EventLocation{Little Whinging, United Kingdom}
\EventLogo{}
\SeriesVolume{42}
\ArticleNo{23}

\definecolor{darkgreen}{RGB}{0,128,0}
\definecolor{darkred}{RGB}{128,0,0}

\newcommand{\R}{\mathbb{R}}
\newcommand{\N}{\mathbb{N}}
\newcommand{\Z}{\mathbb{Z}}
\newcommand{\I}{\mathcal{I}}
\newcommand{\U}{\mathcal{U}}
\newcommand{\F}{\mathcal{F}}
\newcommand{\sol}{\mathcal{S}}
\newcommand{\powerset}[1]{2^{#1}}
\newcommand{\set}[1]{\{ #1 \}}
\newcommand{\fromto}[2]{\set{#1, \ldots, #2}}

\DeclareMathOperator{\poly}{poly}

\newcommand{\leqSSP}{\leq_\text{SSP}}
\newcommand{\bin}{\set{0,1}}

\begin{document}

\maketitle
\begin{abstract}
We consider the general problem of blocking all solutions of some given combinatorial problem with only few elements. 
For example, the problem of destroying all Hamiltonian cycles of a given graph by forbidding only few edges; or the problem of destroying all maximum cliques of a given graph by forbidding only few vertices.
Problems of this kind are so fundamental that they have been studied under many different names in many different disjoint research communities already since the 90s.
Depending on the context, they have been called the interdiction, most vital vertex, most vital edge, blocker, or vertex deletion problem.

Despite their apparent popularity, surprisingly little is known about the computational complexity of interdiction problems in the case where the original problem is already NP-complete.
In this paper, we fill that gap of knowledge by showing that a large amount of interdiction problems are even harder than NP-hard. 
Namely, they are complete for the second stage of Stockmeyer's polynomial hierarchy, the complexity class $\Sigma^p_2$.
Such complexity insights are important because they imply that all these problems can not be modelled by a compact integer program (unless the unlikely conjecture NP $= \Sigma_2^p$ holds).
Concretely, we prove $\Sigma^p_2$-completeness of the following interdiction problems:
    satisfiability,
    3satisfiability,
    dominating set,
    set cover,
    hitting set,
    feedback vertex set,
    feedback arc set,
    uncapacitated facility location,
    $p$-center,
    $p$-median,
    independent set,
    clique,
    subset sum,
    knapsack,
    Hamiltonian path/cycle (directed/undirected),
    TSP,
    $k$ directed vertex disjoint path ($k \geq 2$),
    Steiner tree.
We show that all of these problems share an abstract property which implies that their interdiction counterpart is $\Sigma_2^p$-complete.
Thus, all of these problems are $\Sigma_2^p$-complete \enquote{for the same reason}.
Our result extends a recent framework by Grüne and Wulf.
    
\end{abstract}

\section{Introduction}

This paper is concerned with the \emph{minimum cardinality interdiction problem}, by which we understand the following task: Given some base problem (the so-called \emph{nominal problem}) we wish to find a small subset of elements such that this subset has a non-empty intersection with every optimal solution of the base problem.
The concept of interdiction is so natural that is has re-appeared under many different names in different research communities. 
Depending on the context, the interdiction problem (or slight variants of it) has been called the \emph{most vital node/most vital edge} problem, the \emph{blocker} problem, and \emph{node deletion/edge deletion} problem. (More details on the sometimes subtle differences between these variants is provided further below.)
As an example for the type of problems that this paper is concerned with, consider the following two problems:
\begin{quote}
    \textbf{Problem}: {\sc Min Cardinality Clique Interdiction}

    \textbf{Input}: Graph $G = (V, E)$

    \textbf{Task}: Find a minimum-size subset $V' \subseteq V$ such that every maximum clique shares at least one vertex with $V'$. 
\end{quote}

\begin{quote}
    \textbf{Problem}: {\sc Min Cardinality Hamiltonian Cycle Interdiction}

    \textbf{Input}: Graph $G = (V, E)$

    \textbf{Task}: Find a minimum-size subset $E' \subseteq E$ such that every Hamiltonian cycle shares at least one edge with $E'$.
\end{quote}

In particular, if in the above examples the set $V'$ (respectively the set $E'$) is deleted from the graph, 
the maximum clique size decreases (respectively the graph becomes Hamiltonian-cycle-free). 
Hence the interdiction problem can be interpreted as the minimal effort required to destroy all optimal solutions.
Clearly, analogous problems can be defined and analyzed for a wealth of different nominal problems. Indeed, this has been done extensively by past researchers. 
The following is a non-exhaustive list:
Interdiction-like problems have been considered already since the 90's for a large amount of problems, among others for
shortest path \cite{bar1998complexity,DBLP:journals/mst/KhachiyanBBEGRZ08,malik1989k},
matching \cite{DBLP:journals/dam/Zenklusen10a},
minimum spanning tree \cite{DBLP:journals/ipl/LinC93},
or maximum flow \cite{WOOD19931}.
Note that in all these cases the nominal problem can be solved in polynomial time. Interdiction for nominal problems that are NP-complete has also been extensively considered, for example for
vertex covers \cite{DBLP:conf/iwoca/BazganTT10, DBLP:journals/dam/BazganTT11},
independent sets \cite{DBLP:journals/gc/BazganBPR15,DBLP:conf/iwoca/BazganTT10, DBLP:journals/dam/BazganTT11,DBLP:journals/dam/HoangLW23, DBLP:conf/tamc/PaulusmaPR17},
colorings \cite{DBLP:journals/gc/BazganBPR15, DBLP:conf/iscopt/PaulusmaPR16, DBLP:conf/tamc/PaulusmaPR17},
cliques \cite{DBLP:journals/eor/FuriniLMS19, DBLP:journals/anor/Pajouh20, DBLP:journals/networks/PajouhBP14,DBLP:conf/iscopt/PaulusmaPR16},
knapsack \cite{weninger2024fast,DBLP:conf/ipco/CapraraCLW13},
dominating sets \cite{DBLP:journals/eor/PajouhWBP15},
facility location \cite{DBLP:journals/tcs/FrohlichR21},
1- and $p$-center \cite{DBLP:conf/cocoa/BazganTV10, DBLP:journals/jco/BazganTV13}, and
1- and $p$-median \cite{DBLP:conf/cocoa/BazganTV10, DBLP:journals/jco/BazganTV13}.
A general survey is provided by Smith, Prince and Geunes \cite{smith2013modern}.

This large interest is due to the fact that interdiction problems are well-motivated from many different directions. 
In the area of robust optimization, interdiction is studied because it concerns robust network design, defense against (terrorist) attacks, and sensitivity analysis \cite{DBLP:journals/corr/abs-2406-01756}. 
In particular, we want to find the most vital nodes/edges of a given network in order to identify its most vulnerable points, and understand where small changes have the largest impact. 
Interdiction in these contexts is often interpreted as a min-max optimization problem, or alternatively as a game between a network interdictor (attacker) and a network owner (defender) with competing goals.
In the area of bilevel optimization, interdiction-like problems arise naturally from the dynamic between two independent hierarchical agents \cite{DBLP:conf/ipco/CapraraCLW13}.
In the area of pure graph theory, interdiction problems are usually called vertex and edge blocker problems. 
They relate to the important concepts of maximum induced subgraphs, critical vertices and edges, cores, and transversals (with respect to some fixed property) \cite{DBLP:conf/tamc/PaulusmaPR17}.
In the area of (parameterized) complexity, interdiction-like problems are usually called vertex deletion problems.
They arise from the desire to delete a constant number of vertices until the resulting graph has some desirable property, for example so that it can be handled by an efficient algorithm.
For instance, Lewis and Yannakakis showed that the vertex deletion problem for hereditary graph properties is NP-complete \cite{DBLP:journals/jcss/LewisY80} and Bannach, Chudigiewitsch and Tantau analyzed the parameterized complexity for properties formulatable by first order formulas \cite{DBLP:conf/mfcs/BannachCT24}.

\textbf{The natural complexity of minimum cardinality interdiction.}
In this paper, we are mainly concerned with interdiction problems where the nominal problem is already NP-complete. From a complexity-theoretic point of view, such interdiction problems are often times even harder than NP-complete, namely they are complete for the second stage in the so-called \emph{polynomial hierarchy} \cite{DBLP:journals/tcs/Stockmeyer76}. A problem complete for the second stage of the hierarchy is called $\Sigma^p_2$-complete. 
The theoretical study of $\Sigma^p_2$-complete problems is important:
If a problem is found to be $\Sigma^p_2$-complete, it means that, under some basic complexity-theoretic assumptions\footnote{More specifically, we assume here that $\text{NP} \neq \Sigma^p_2$, i.e.\ we assume the polynpmial hierarchy does not collapse to the first level. Similar to the famous $\text{P} \neq \text{NP}$ conjecture, this is believed to be unlikely by experts.
However, the true status of the conjecture is not known
(see e.g.\ \cite{DBLP:journals/4or/Woeginger21}).}, it is not possible to find a mixed-integer programming formulation of the problem of polynomial size \cite{DBLP:journals/4or/Woeginger21} (also called a \emph{compact} model).
This means that no matter how cleverly a decision maker tries to design their mixed integer programming model, it must inherently have a huge number of constraints and/or variables, and may be much harder to solve than even NP-complete problems. 
Furthermore, for the type of interdiction problems discussed here, where the nominal problem is NP-complete, 
under the same assumption $\text{NP} \neq \Sigma^p_2$, one can show that (the decision variant of) the interdiction problem is often times actually not contained in the complexity class NP, only in the class $\Sigma^p_2$.
Hence the class $\Sigma^p_2$ is the natural class for this type of problem.

Even though this fact makes the study of $\Sigma^p_2$-complete problems compelling, and even though interdiction-like problems have received a large amount of attention in recent years, 
surprisingly few $\Sigma^p_2$-completeness results relevant to the area of interdiction were known until recently. 
While the usual approach to prove $\Sigma^p_2$-completeness (or NP-completeness) is to formulate a new proof for each single problem,
a recent paper by Grüne \& Wulf \cite{gruene2024completeness}, extending earlier ideas by Johannes \cite{johannes2011new} breaks with this approach. 
Instead, it is shown that there exists a large 'base list' of problems (called \emph{SSP-NP-complete} problems in \cite{gruene2024completeness}), including many classic problems like satisfiability, vertex cover, clique, knapsack, subset sum, Hamiltonian cycle, etc.
Grüne and Wulf show that for each problem from the base list, some corresponding min-max version is $\Sigma^p_2$-complete, and some corresponding min-max-min version is $\Sigma^p_3$-complete. 
This approach has three main advantages: 
1.) It uncovers a large number of previously unknown $\Sigma^p_2$-complete problems. 
2.) It reveals the theoretically interesting insight, that for all these problems the $\Sigma^p_2$-completeness follows from essentially the same argument. 
3.) It can simplify future proofs, since heuristically it seems to be true that for a new problem it is often easier to show that the nominal problem belongs to the list of SSP-NP-complete problems, than to find a $\Sigma^p_2$-completeness proof from scratch.

\textbf{Our results.}
In this paper, we extend the framework of Grüne \& Wulf \cite{gruene2024completeness} to include the case of minimum cardinality interdiction problems. 
We remark that the original framework of Grüne \& Wulf already shows such a result in the case where the action of interdicting an element is associated with so-called interdiction costs, which may be different for each element. 
Hence our work can be understood as an extension to the unit-cost case, which is arguably the most natural variant of interdiction.
Concretely, in this paper we consider minimum cardinality interdiction simultaneously for all of the following nominal problems:
\begin{quote}
    satisfiability,
    3satisfiability,
    dominating set,
    set cover,
    hitting set,
    feedback vertex set,
    feedback arc set,
    uncapacitated facility location,
    $p$-center,
    $p$-median,
    independent set,
    clique,
    subset sum,
    knapsack,
    Hamiltonian path/cycle (directed/undirected),
    TSP,
    $k$ directed vertex disjoint path ($k \geq 2$),
    Steiner tree.
\end{quote}
We show that for all these problems, the minimum cardinality interdiction problem is $\Sigma^p_2$-complete.

More abstractly, we introduce a meta-theorem from which our concrete results follows.
This means we introduce a set of sufficient conditions for some nominal problem, 
which imply that the minimum cardinality interdiction problem becomes $\Sigma^p_2$-complete.
It turns out that compared to the original framework of Grüne and Wulf, additional assumptions are necessary in the unit-cost case. 

We remark that $\Sigma^p_2$-completeness was already known in the case of clique/independent set, and knapsack \cite{DBLP:conf/ipco/CapraraCLW13,DBLP:journals/amai/Rutenburg93,DBLP:journals/corr/abs-2406-01756}. Hence our work is an extension of these results.

\textbf{Related Work}. Usually in the literature, the complexity of interdiction problems is not discussed beyond NP-hardness. 
However, there are the following exceptions: Rutenburg \cite{DBLP:journals/amai/Rutenburg93} proves $\Sigma^p_2$-completeness for clique interdiction. Caprara, Carvalho, Lodi \& Woeginger \cite{DBLP:conf/ipco/CapraraCLW13} consider different bilevel knapsack formulations and prove $\Sigma^p_2$-completeness of the DeNegre \cite{10.5555/2231641} knapsack variant, which can be interpreted as an interdiction knapsack variant.
Tomasaz, Carvalho, Cordone \& Hosteins \cite{DBLP:journals/corr/abs-2406-01756} consider interdiction-fortification games and prove $\Sigma^p_2$-completeness of another knapsack interdiction variant. 
Fröhlich and Ruzika prove $\Sigma^p_2$-completeness of a facility location interdiction problem on graphs
(in contrast to our work, the interdictor attacks edges instead of vertices) \cite[Section 4]{DBLP:journals/tcs/FrohlichR21}.
Our work extends these results to more problem classes.
Finally, in a seminal paper, Lewis \& Yannakakis prove the very general result that the most vital vertex problem is NP-hard for every nontrivial hereditary graph property \cite{DBLP:journals/jcss/LewisY80}.
Our work adds to these results by showing that in many cases, interdiction is even harder than NP-hard.
As already mentioned, our work is based on the framework by Grüne and Wulf \cite{gruene2024completeness}, which itself is based on earlier ideas by Johannes \cite{johannes2011new}.

\section{Preliminaries}
\label{sec:prelim}

A \emph{language} is a set $L\subseteq \bin^*$.
A language $L$ is contained in $\Sigma^p_k$ iff there exists some polynomial-time computable function $V$ (verifier), and $m_1,m_2,\ldots, m_k = \poly(|w|)$ such that for all $w \in \set{0,1}^*$
\[
    w \in L \ \Leftrightarrow \ \exists y_1 \in \set{0,1}^{m_1} \ \forall y_2 \in \set{0,1}^{m_2} \ldots \ Q y_k \in \set{0,1}^{m_k}: V(w,y_1,y_2,\ldots,y_k) = 1,
\]
where $Q = \exists$, if $k$ is odd, and $Q = \forall$, if $k$ even.

An introduction to the polynomial hierarchy and the classes $\Sigma^p_k$ can be found in the book by Papadimitriou \cite{DBLP:books/daglib/0072413} or in the article by Jeroslow \cite{DBLP:journals/mp/Jeroslow85}.
An introduction specifically in the context of bilevel optimization can be found in the article of Woeginger \cite{DBLP:journals/4or/Woeginger21}.

A \emph{many-one-reduction} or \emph{Karp-reduction} from a language $L$ to a language $L'$ is a map $f : \bin^* \to \bin^*$ such that $w \in L$ iff $f(w) \in L'$ for all $w \in \bin^*$. 
A language $L$ is $\Sigma^p_k$-hard, if every $L' \in \Sigma^p_k$ can be reduced to $L$ with a polynomial-time many-one reduction. If $L$ is both $\Sigma^p_k$-hard and contained in $\Sigma^p_k$, it is $\Sigma^p_k$-complete.

For some cost function $c : U \to \R$, and some subset $U' \subseteq U$, we define the cost of the subset $U'$ as $c(U') := \sum_{u \in U'} c(u)$. For a map $f : A \to B$ and some subset $A' \subseteq A$, we define the image of the subset $A'$ as $f(A') = \set{f(a) : a \in A'}$.
\section{Framework}
\label{sec:framework} 

Since this work is an extension of the framework of Grüne \& Wulf, it becomes necessary to re-introduce the most important concepts of the framework. 
A more in-depth explanation of these concepts and their motivation in provided in the original paper \cite{gruene2024completeness}.
Grüne \& Wulf start by giving a precise definition of the objects they are interested in, linear optimization problems (LOP).
An example of an LOP problem is the vertex cover problem.

\begin{definition}[Linear Optimization Problem, from \cite{gruene2024completeness}]
\label{def:LOSPP}
    A linear optimization problem (or in short LOP)  $\Pi$ is a tuple $(\I, \U, \F, d, t)$, such that
    \begin{itemize}
        \item $\I \subseteq \{0,1\}^*$ is a language. We call $\I$ the set of instances of $\Pi$.
        \item To each instance $I \in \I$, there is some
        \begin{itemize}
            \item set $\U(I)$ which we call the universe associated to the instance $I$.
            \item set $\F(I) \subseteq \powerset{\U(I)}$ that we call the feasible solution set associated to the instance $I$. 
            \item function $d^{(I)}: \U(I) \rightarrow \Z$ mapping each universe element $e$ to its costs $d^{(I)}(e)$.
            \item threshold $t^{(I)} \in \Z$. 
        \end{itemize}
    \end{itemize}
    For $I \in \I$, we define the solution set $\sol(I) := \set{S \in \F(I) : d^{(I)}(S) \leq t^{(I)}}$ as the set of feasible solutions below the cost threshold. 
    The instance $I$ is a Yes-instance, if and only if $\sol(I) \neq \emptyset$.
    We assume (for LOP problems in NP) that it can be checked in polynomial time in $|I|$ whether some proposed set $F \subseteq \U(I)$ is feasible.
\end{definition}

\begin{description}
    \item[]\textsc{Vertex Cover}\hfill\\
    \textbf{Instances:} Graph $G = (V, E)$, number $k \in \N$.\\
    \textbf{Universe:} Vertex set $V =: \U$.\\
    \textbf{Feasible solution set:} The set of all vertex covers of $G$.\\
    \textbf{Solution set:} The set of all vertex covers of $G$ of size at most $k$.
\end{description}

It turns out that often times the mathematical discussion is a lot clearer, when one omits the concepts $\F, d^{(I)}$, and $t^{(I)}$, since for the abstract proof of the theorems only $\I$, $\U$, $\sol$ are important.
This leads to the following abstraction from the concept of an LOP problem:

\begin{definition}[Subset Search Problem (SSP), from \cite{gruene2024completeness}]
\label{def:SSP}
A subset search problem (or short SSP problem) $\Pi$ is a tuple $(\I, \U, \sol)$, such that
\begin{itemize}
    \item $\I \subseteq \set{0,1}^*$ is a language. We call $\I$ the set of instances of $\Pi$. 
    \item To each instance $I \in \I$, there is some set $\U(I)$ which we call the universe associated to the instance $I$. 
    \item To each instance $I \in \I$, there is some (potentially empty) set $\sol(I)\subseteq \powerset{\U(I)}$ which we call the solution set associated to the instance $I$.
\end{itemize}
\end{definition}

An instance of an SSP problem is a called yes-instance, if $\sol(I) \neq \emptyset$. 
Every LOP problem becomes an SSP problem with the definition $\sol(I) := \set{S \in \F(I) : d^{(I)}(S) \leq t^{(I)}}$.
We call this the \emph{SSP problem derived from an LOP problem}. Some problems are more naturally modelled as an SSP problem to begin with, rather than as an LOP problem.
For example, the satisfiability problem becomes an SSP problem with the following definition.

\begin{description}
    \item[]\textsc{Satisfiability}\hfill\\
    \textbf{Instances:} Literal set $L = \fromto{\ell_1}{\ell_n} \cup \fromto{\overline \ell_1}{\overline \ell_n}$, clause set $C = \fromto{C_1}{C_m}$ such that $C_j \subseteq L$ for all $j \in \fromto{1}{m}.$\\
    \textbf{Universe:} $L =: \U$.\\
    \textbf{Solution set:} The set of all subsets $L' \subseteq \U$ of the literals such that for all $i \in \fromto{1}{n}$ we have $|L' \cap \set{\ell_i, \overline \ell_i}| = 1$, and such that $|L' \cap C_j| \geq 1$ for all clauses $C_j \in C$.
\end{description}

Grüne \& Wulf introduce a new type of reduction, called \emph{SSP reduction}. 
Roughly speaking, a usual polynomial-time reduction from some problem $\Pi$ to another problem $\Pi'$ has the SSP property, 
if it comes with an additional injective map $f$ which embeds the universe of $\Pi$ into the universe of $\Pi'$, 
in such a way that $\Pi$ can be interpreted as a \enquote{subinstance} of $\Pi'$ and the topology of solutions is maintained in the subset that is induced by the image of $f$. 
More formally, let $W$ denote the image of $f$. We interpret $W$ as the subinstance of $\Pi$ contained in the instance of $\Pi'$ and we want the following two conditions to hold: 
1.) For every solution $S'$ of $\Pi'$, the set $f^{-1}(S' \cap W)$ is a solution of $\Pi$. 
2.) For, every solution $S$ of $\Pi$, the set $f(S)$ is a partial solution of $\Pi'$ and can be completed to a solution by using elements not in $W$.
These two conditions together are summarized in the single equation (\ref{eq:SSP}). We write $\Pi \leq_\text{SSP} \Pi'$ to denote that such a reduction exists.
We refer the reader to \cite{gruene2024completeness} for a more intuitive explanation of these properties and an example 3\textsc{Sat} $\leq_\text{SSP}$ \textsc{vertex cover}.

\begin{definition}[SSP Reduction, from \cite{gruene2024completeness}]
\label{def:ssp-reduction}
    Let $\Pi = (\I,\U,\sol)$ and $\Pi' = (\I',\U',\sol')$ be two SSP problems. We say that there is an SSP reduction from $\Pi$ to $\Pi'$, and write $\Pi \leqSSP \Pi'$, if
    \begin{itemize}
        \item There exists a function $g : \I \to \I'$ computable in polynomial time in the input size $|I|$, such that $I$ is a Yes-instance iff $g(I)$ is a Yes-instance (i.e. $\sol(I) \neq \emptyset$ iff $\sol'(g(I)) \neq \emptyset$).
        \item There exist functions $(f_I)_{I \in \I}$ computable in polynomial time in $|I|$ such that for all instances $I \in \I$, we have that $f_I : \U(I) \to \U'(g(I))$ is an injective function mapping from the universe of the instance $I$ to the universe of the instance $g(I)$ such that 
        \begin{equation}
            \set{f_I(S) : S \in \sol(I) } = \set{S' \cap f_I(\U(I)) : S' \in  \sol'(g(I))}. \label{eq:SSP}
        \end{equation}

    \end{itemize}
\end{definition}

It is shown in \cite{gruene2024completeness} that SSP reductions are transitive, i.e.\ $\Pi_1 \leqSSP \Pi_2$ and $\Pi_2 \leqSSP \Pi_3$ implies $\Pi_1 \leqSSP \Pi_3$.
The class of SSP-NP-complete problems is denoted by SSP-NPc and consists out of all SSP problems $\Pi$ that are polynomially-time verifiable and such that $\textsc{Satisfiability} \leq_\text{SSP} \Pi$. 
The main observation in \cite{gruene2024completeness} is that many classic problems are contained in the class SSP-NPc, 
and that this fact can be used to prove that their corresponding min-max versions are $\Sigma^p_2$-complete. 

\section{Minimum Cardinality Interdiction Problems}

In this section, we prove our $\Sigma^p_2$-completeness results regarding the minimum cardinality interdiction problem. 
Since we want to prove the theorem simultaneously for multiple problems at once, we require an abstract definition of the interdiction problem.
For this, consider the following definition.

\begin{definition}[Minimum Cardinality Interdiction Problem]
\label{def:min-card-interdiction-pi}
    Let an SSP problem $\Pi = (\I, \U, \sol)$ be given.
    The minimum cardinality interdiction problem associated to $\Pi$ is denoted by \textsc{Min Cardinality Interdiction-$\Pi$} and defined as follows:
    The input is an instance $I \in \I$ together with a number $k \in \N_0$.
    The question is whether
    \[
        \exists B \subseteq \U(I), \ |B| \leq k : \forall S \in \sol(I) : B \cap S \neq \emptyset.
    \]
\end{definition}

For the remainder of the paper, it is helpful to imagine this problem as a game between two players: the \emph{attacker} and the \emph{defender}.
That is, interdiction is an action performed by an attacker (or interdictor), who wishes to select a blocker of few elements to destroy all solutions.
On the other hand, the defender wants to find a solution to the problem after the attacker selected a blocker.
This leads to the following interpretation:
\begin{itemize}
    \item The set $\U(I)$ contains all the elements the attacker is allowed to attack. 
    \item The set $\sol(I)$ contains all the solutions the attacker wants to destroy such that the defender is not able to find any solution.
    For example, this could be the set of all Hamiltonian cycles, the set of all cliques of a certain size, etc.
\end{itemize}
Therefore, the formulation of the base problem as SSP problem $(\I, \U, \sol)$ determines which elements the attacker can attack, which he cannot attack (e.g. edges/vertices of a graph), and what the attacker's goal is.
We note that different formulations $(\I, \U, \sol)$ of the same problem are formally different SSP problems. They might be both SSP-NP-complete independent of each other, but require their own SSP-NP-completeness proof each.
For all the concrete problems studied in this paper, our complexity results hold for the natural choices of $(\I, \U, \sol)$ formally given in \cref{app:sec:problemDefinitions}.
Finally, note that if the base problem is an LOP problem, then by definition $\sol(I)$ is the set of feasible solutions below some threshold specified in the input. 
For example, applying \cref{def:min-card-interdiction-pi} to $\Pi = \textsc{Clique}$ yields the following decision problem: 
\begin{quote}
    \textbf{Problem}: $\textsc{Min Cardinality Interdiction-Clique}$

    \textbf{Input}: Graph $G = (V, E)$, numbers $k, t \in\N_0$

    \textbf{Question}: Does there exist a subset $B \subseteq V$ of size $|B| \leq k$ such that every clique of size at least $t$ shares at least one vertex with $B$? 
\end{quote}

Some more technical details, concerning the subtle differences between different variants of interdiction referenced in the literature as well as concerning the question whether $t$ can be chosen to be optimal are discussed in \cref{sec:different-variants-of-interdiction}.
We now proceed with the main result.
For the complexity analysis of minimum cardinality blocker, we first show the containment in the class $\Sigma^p_2$, if the nominal problem is in NP.

\begin{lemma}
\label{lem:containment}
    Let $\Pi = (\I, \U, \sol)$ be an SSP problem in $NP$, then \textsc{Min Cardinality Interdiction}-$\Pi$ is in $\Sigma^p_2$.
\end{lemma}
\begin{proof}
    We provide a polynomial time algorithm $V$ that verifies a specific solution $y_1, y_2$ of polynomial size for instance $I$ such that
    $$
        I \in L \ \Leftrightarrow \ \exists y_1 \in \{0,1\}^{m_1} \ \forall y_2 \in \{0,1\}^{m_2} : V(I, y_1, y_2) = 1.
    $$
    With the $\exists$-quantified $y_1$, we encode the blocker $B \subseteq \U(I)$.
    The encoding size of $y_1$ is polynomially bounded in the input size of $\Pi$ because $|\U(I)| \leq poly(|x|)$.
    Next, we encode the solution $S \in \sol(I)$ to the nominal problem $\Pi$ using the $\forall$-quantified $y_2$ within polynomial space.
    This is doable because the problem $\Pi$ is in NP (and thus $co\Pi$ is in coNP).
    At last, the verifier $V$ has to verify the correctness of the given solution provided by the $\exists$-quantified $y_1$ and $\forall$-quantified $y_2$.
    Checking whether $|B| \leq t$ and $B \cap S \neq \emptyset$ is trivial and checking whether $S \in \sol(I)$ is clearly in polynomial time because $\Pi$ is in SSP-NP.
    It follows that $\textsc{Min Cardinality Interdiction-}\Pi$ is in $\Sigma^p_2$.
\end{proof}

Next, we show the hardness of minimum cardinality interdiction problems as long as the nominal problem is NP-complete.
For this, we introduce the concept of invulnerability reductions that helps us to grasp the problems in a unified approach.
We describe this concept in the following subsection with the goal to obtain the following main theorem of the paper.

\begin{theorem}
\label{thm:min-card-interdiction}
    The problem \textsc{Min Cardinality interdiction-$\Pi$} is $\Sigma^p_2$-complete for all the following problems:
    independent set,
    clique,
    subset sum,
    knapsack,
    Hamiltonian path/cycle (directed/undirected),
    TSP,
    $k$-directed vertex disjoint paths ($k \geq 2$),
    Steiner tree,
    dominating set,
    set cover,
    hitting set, feedback vertex set,
    feedback arc set,
    uncapacitated facility location,
    $p$-center,
    $p$-median.
\end{theorem}

We remark that the case of satisfiability deserves special attention, which is discussed more thoroughly in \cref{sec:noMeta}.

\subsection{Invulnerability Reduction}

Our proof strategy for each of the problems listed in \Cref{thm:min-card-interdiction} is essentially the same.
In fact, we show that \cref{thm:min-card-interdiction} is actually a consequence of the following, more powerful \emph{meta-theorem}.
This meta-theorem catches the essence of an invulnarability reduction.

\begin{theorem}
\label{thm:meta-theorem}
    Consider an SSP-NP-complete problem $\Pi$.
    If there exists a polynomial-time reduction $g$ which receives as input a tuple $(I, C, k)$ of an instance $I$ of $\Pi$, some set $C \subseteq \U(I)$ and some $k \in \N_0$, and returns instances $I' := g(I, C, k)$ of $\Pi$, such that the following holds: 
     \begin{align*}
         \exists B \subseteq C : |B| \leq k \text{ and } B \cap S \neq \emptyset \ \forall S \in \sol(I) \qquad\qquad\qquad\qquad\qquad\qquad\qquad\qquad\quad\\
         \Leftrightarrow \ \ \exists B' \subseteq \U(I') : |B'| \leq k \text{ and } B' \cap S' \neq \emptyset \ \forall S' \in \sol(I').
     \end{align*}
     Then \textsc{Min Cardinality Interdiction-$\Pi$} is $\Sigma^p_2$-complete. 
\end{theorem}

It would be nice to have \cref{thm:meta-theorem} for all problems in the class SSP-NPc, not only those who admit a funciton $g$ with the properties as described above. However, we give a reasoning in \Cref{sec:noMeta} why such a generalization is not possible.
The rest of this section is devoted to the proof of \cref{thm:min-card-interdiction}.
In \cite{gruene2024completeness} the following more general version of interdiction was considered, where there is a set $C \subseteq \U(I)$ of so-called vulnerable elements.
One can also interpret the set of vulnerable elements $C$ as the elements that have cost of interdiction of $1$ while all other elements $\U(I) \setminus C$ have a cost of interdiction of $\infty$ and a blocker of small costs is sought.
This problem is called the \emph{combinatorial interdiction problem}.

\begin{definition}[Comb. Interdiction Problem, from \cite{gruene2024completeness}.]
    Let an SSP problem $\Pi = (\I, \U, \sol)$ be given.
    We define \textsc{Comb. Interdiction-$\Pi$} as follows:
    The input is an instance $I \in \I$, a number $k \in \N_0$, and a set $C \subseteq \U(I)$. The set $C$ is called the set of vulnerable elements.
    The question is whether
    \[
        \exists B \subseteq C, \ |B| \leq k : \forall S \in \sol(I) : B \cap S \neq \emptyset.
    \]
\end{definition}

It is proven in \cite{gruene2024completeness} that for every problem in SSP-NPc, the combinatorial interdiction problem is $\Sigma^p_2$-complete.
Now, let $\Pi$ be in SSP-NPc and $g$ be a reduction such that
    \begin{align*}
         \exists B \subseteq C : |B| \leq k \text{ and } B \cap S \neq \emptyset \ \forall S \in \sol(I) \qquad\qquad\qquad\qquad\qquad\qquad\qquad\qquad\quad\\
         \Leftrightarrow \ \ \exists B' \subseteq \U(I') : |B'| \leq k \text{ and } B' \cap S' \neq \emptyset \ \forall S' \in \sol(I'),
     \end{align*}
then $g$ is a reduction from \textsc{Comb. Interdiction-$\Pi$} to \textsc{Min Cardinality interdiction-$\Pi$}. This is because the first line is equivalent to the statement that instance $I$ is a yes-instance of \textsc{Comb. Interdiction-$\Pi$}, and the second line is equivalent to the statement that $I'$ is a yes-instance of \textsc{Min Cardinality interdiction-$\Pi$}.
It directly follows that \textsc{Min Cardinality interdiction-$\Pi$} is $\Sigma^p_2$-complete. This completes the proof of \cref{thm:meta-theorem}. 

We remark that while in some sense the proof is rather trivial, we still see a lot of value in explicitly stating a set of easy-to-check sufficient conditions that render some minimum-cardinality interdiction problem $\Sigma^p_2$-complete.

How can one find a function $g$ with the properties as described above? Often times it is possible by employing the following natural idea:
Given an instance of the comb. interdiction problem, let the set $D := \U(I) \setminus C$ be called the \emph{invulnerable} elements. 
For each problem separately we explain that a gadget for the invulnerable elements in $D$ exists, which
intuitively speaking guarantees that an attacker, no matter which $k$ elements of the universe they attack, can never render the elements of $D$ unusable.
On the other hand, we make sure that the \emph{invulnerability gadgets} do not meaningfully change the set of solutions.
The next section gives many examples of such gadgets.
We remark that we are not the first to come up with this natural idea.
For example, Zenklusen \cite{DBLP:journals/dam/Zenklusen10a} used the same idea in the context of matching interdiction.

\subsection{Different Variants of Interdiction}
\label{sec:different-variants-of-interdiction}
In this section, we discuss variants of interdiction problems that can be found in the literature.
For this, we study the relation of our definition of a minimum cardinality interdiction problems and the existing variants.
Additionally, we argue what the implications of the hardness of our minimum cardinality interdiction problems on the other variants are.

\begin{description}
    \item[1. Minimal Blocker Problem.]\hfill
        \begin{description}
            \item[Input] Instance $I$ with universe $U$, blocker cost function $c$, solution cost function $d$, and solution threshold $\tau$
            \item[Task] Find the minimum-cost set $\min_{B \subseteq U} c(B)$ such that for all solutions $S$ with $S \cap B = \emptyset$, we have $d(S) \leq \tau$.
        \end{description}
    \item[2. Full Decision Variant of Interdiction.]\hfill
        \begin{description}
            \item[Input] Instance $I$ with universe $U$, blocker cost function $c$, blocker budget $k$, solution cost function $d$, and solution threshold $\tau$
            \item[Task] Is there a set $B \subseteq U$ with $c(B) \leq k$ such that for all solutions $S$ with $S \cap B = \emptyset$, we have $d(S) \leq \tau$?
        \end{description}
    \item[3. Most Vital Elements Problem.]\hfill
        \begin{description}
            \item[Input] Instance $I$ with universe $U$, blocker cost function $c$, and solution cost function $d$
            \item[Task] Find a set $B \subseteq U$ with $c(B) \leq k$ such that the costs of all solutions $S \cap B = \emptyset$ are maximized, i.e. $\max_{B} \min_{S, S \cap B = \emptyset} d(S)$.
        \end{description}
\end{description}

Our goal is to show that all of the variants from above are at least as hard as our formulation of \emph{minimum cardinality interdiction} (\Cref{def:min-card-interdiction-pi}).
This results in the following theorem.

\begin{theorem}
    Let $\Pi = (\I, \U, \sol)$ be an SSP problem.
    Then the
    Most Vital Elements Problem of $\Pi$ (for all problems $\Pi$ in \Cref{thm:min-card-interdiction}),
    the Minimal Blocker Problem of $\Pi$, and
    the Full Decision Variant of Interdiction of $\Pi$
    are at least as hard to compute as \textsc{Min Cardinality Interdiction}-$\Pi$.
\end{theorem}

The rest of this section is devoted to the proof of this theorem.
In our formulation of minimum cardinality interdiction, a set $B$ is sought, which intersects every solution in the set $\sol$ as given by the corresponding SSP problem.
We now have to distinguish between problems, which are naturally formulated as SSP problems (e.g. Hamiltonian cycle), and SSP problems, which are derived from an LOP problem (e.g. clique).
For natural SSP problems, the solution set $\sol$ consists of all solutions, i.e. there are no feasible solutions outside of $\sol$ due to the missing cost function $d$ on the solution elements.
Thus all of the three variants from above are generalizations of minimum cardinality interdiction:
\begin{enumerate}
    \item The \emph{minimal blocker problem} is the optimization version of the corresponding minimum cardinality interdiction problem.
    \item The \emph{full decision version of interdiction} is a generalization of the corresponding minimum cardinality interdiction problem because the latter assumes to have unit costs in the cost function $c$ for all elements from $U$.
    \item The \emph{most vital element problem} behaves analogous to (2).
\end{enumerate}
For SSP problems that are derived from an LOP problem, basically the same holds, however, with a modified and a technically more intricate argumentation.
Here the solution set is defined by $\sol = \{F \in \F : d(F) \leq t\}$ and we can find a reduction by generalization as follows:
\begin{enumerate}
    \item For \emph{minimal blocker problems}, we can set $\tau := t-1$.
    Then, we again have that the minimal blocker problem is the optimization version of the corresponding minimum cardinality interdiction problem.
    \item For \emph{minimal blocker problems}, we can also set $\tau := t-1$.
    Then, the full decision version is again a generalization of the corresponding minimum cardinality interdiction problem due to the fact that the latter has a unit cost function $c$.
    \item For \emph{most vital element problem}, the situation is more complicated.
    We first observe that the blocker part of $B \subseteq U$ with $c(B) \leq k$ is a generalization of the blocker part in minimum cardinality interdiction.
    The inner part on the nominal problem deserves special attention, though, due to the fact that the most vital element problem maximizes the objective while minimum cardinality interdiction blocks all solutions from the solution set $\sol$.
    We focus on this in the next paragraph.
\end{enumerate}

\par{\bf Reducing Minimum Cardinality Interdiction to Most Vital Elements.}
The concepts of minimum cardinality interdiction and most vital elements coincide if and only of the set $\sol$ contains exactly the optimal solutions, i.e. $\sol = \{F \in \F : d(F) \leq t^\star\}$, where $t^\star$ is optimal (i.e. minimal).
In order to assure that $\sol$ captures exactly the optimal solutions, we need to include this condition into the reduction.
In particular, the SSP reduction $(g, f)$ needs to guarantee that all instances $I$ are mapped to instances $g(I)$ such that all possible solutions are necessarily optimal.
In other words, $t$ is the optimal objective value of the LOP instance $g(I)$, since there are no feasible solutions, whose cost is even smaller than $t$.
We call SSP reductions that fulfill this criterion \emph{tight} and formally define them as follows.

\begin{definition}[Tight SSP reduction]
    Let $\Pi_1$ be an SSP problem and $\Pi_2 = (\I, \U, \F, d, t)$ be an LOP problem.
    Consider an SSP reduction $(g, (f_I)_{I \in \I})$ from $\Pi_1$ to (the SSP problem derived from) $\Pi_2$. 
    The reduction is called tight if for all yes-instances $I_1$ of $\Pi_1$, the corresponding instance $I_2 = g(I_1)$ of $\Pi_2$ with the associated parameter $t := t^{(I_2)}$ and associated cost function $d := d^{(I_2)}$, the following holds:
    \begin{align}
        \set{ F \in \F(I_2) : d(F) \leq t } \neq \emptyset \text{ and } \set{ F \in \F(I_2) : d(F) \leq t - 1} = \emptyset
    \end{align}
\end{definition}

All SSP reductions (to SSP problems derived form LOP problems) that can be found in \cite{gruene2024completeness} fulfill this definition and are thus tight.
Therefore, for all LOP problems (independent set, clique, knapsack, TSP, Steiner tree, dominating set, set cover, hitting set, feedback vertex set, feedback arc set), we obtain that the most vital element problem is at least as hard to compute as the minimum cardinality problem.

\paragraph*{Vertex/Edge Deletion Problems}
In this paper, we are concerned with finding a set $B$ such that $B \cap S \neq \emptyset$ for every solution $S$.
Note that this definition is meaningful even if the nominal problem is not graph-based.
However, in the special case where the nominal problem is graph-based, one could also consider a very related notion which is usually called \emph{vertex deletion problem} or \emph{edge deletion problem}.
Here, the question is how many vertices (edges) need to be deleted from the graph until some desired property is met.
Element deletion problems are well-studied in classical complexity theory for hereditary graph properties \cite{DBLP:journals/jcss/LewisY80} and in parameterized complexity theory for properties expressible by first order formulas \cite{DBLP:conf/mfcs/BannachCT24}.
In the general case, element deletion problems are not the same problem as our problem \textsc{Interdiction-$\Pi$}.
This is because for every set of deleted elements, the underlying instance is changed (vertices/edges are removed, which changes the graph). This is not the case for minimum cardinality interdiction problems as defined in this paper.
Thus, it is not possible to transfer the results of minimum cardinality interdiction directly to element deletion problems.
Albeit for the problems of clique and independent set, the $\Sigma^p_2$-completeness results hold for both minimum cardinality interdiction as well as for vertex deletion interdiction 
because for these problems the deletion of a vertex coincides with not taking this vertex into the solution. 
An analogous statement holds for edge deletions for the problems of directed/undirected Hamiltonian cycle/path, $k$-vertex-disjoint path, and Steiner tree.

\section{Invulnerability Reductions for Various Problems}
\label{sec:invulnerability-gadgets}
In this section, we show that a lot of well-known problems satisfy the assumptions of \cref{thm:meta-theorem}, i.e.\ it is possible to construct so-called invulnerability gadgets for them.
Note that this proves \cref{thm:min-card-interdiction}.
(More precisely, it proves the hardness part and the containment part is analogous to \cite{gruene2024completeness}).
Let in the following always $C \subseteq \U(I)$ denote the set of vulnerable elements, let $\U(I) \setminus C$ denote the set of invulnerable elements, and $k$ denote the budget of the attacker.

\textbf{Clique.}
We have $\U = V$ in this case.
For a given graph $G = (V,E)$, and a set $C \subseteq V$, we explain how to make $V \setminus C$ invulnerable.
We obtain a graph $G'$ from $G$ by replacing every vertex $v \in V \setminus C$ with an independent set $X_v$ of size $|X_v| = k+1$. 
For a vertex $v \in C$, we define $X_v := \set{v}$.
For all edges $uv$ in $G$, the new graph $G'$ contains the complete bipartite graph between $X_u$ and $X_v$.
Note that every clique of $G'$ contains at most one vertex from every set $X_v$. Hence the size of a maximum clique is the same in $G$ and $G'$. 
Since for $v \in V \setminus C$, we have $|X_v| = k+1$ and all vertices in $X_v$ have the same neighborhood, the attacker is not able to attack all vertices of $X_v$ at once because its budget of $k$ is too small.
Hence $v$ has been made \enquote{invulnerable}.
Furthermore, for every clique in $G$, we find a corresponding clique in $G'$ that contains at most one vertex from each set $X_v$. 
Together, this implies that an attacker can find a set $B' \subseteq V(G')$ of size $|B'| \leq k$ interdicting all maximum cliques in $G'$ if and only the attacker can find a set $B \subseteq C$ of size $|B| \leq k$ interdicting all maximum cliques of $G$, i.e.\ the assumptions of \cref{thm:meta-theorem} are met.

\textbf{Independent Set.} Analogous to clique in the complement graph.

\textbf{Dominating Set.} We have $\U = V$ in this case. 
To make a vertex $v \in V \setminus C$ invulnerable, we use the same construction as for the clique problem, with the only difference that $X_v$ is a clique instead of an independent set. 
Every optimal dominating set takes at most one vertex from each set $X_v$, but all $k+1$ vertices inside $X_v$ are equivalent. More precisely, they have the same (closed) neighborhood. 
This means for an invulnerable $v \in V \setminus C$, an attacker can not attack all $k+1$ vertices of $X_v$ simultaneously. 
Furthermore, it is easily seen that on the vulnerable vertices, the attacker interdicts all optimal dominating sets in the old graph if and only if the analogous attack interdicts all optimal dominating sets in the new graph.

\textbf{Hitting Set.} In this case, we have some universe $\U$, sets $Y_1,\dots,Y_t \subseteq \U$, and the problem is to find a minimal hitting set $X \subseteq \U$ hitting all the sets $Y_j$, $j =1,\dots,t$. 
To make an element $e \in \U$ invulnerable, simply delete it and replace it by $k + 1$ copies.
We modify the sets such that every set $Y_j$ that contained $e$ now contains the $k+1$ copies of $e$ instead. 
It is clear that all the copies of $e$ hit the same sets as $e$ (i.e.\ taking multiple copies into the hitting set does not offer any advantage).
Furthermore, it is not possible for the attacker to attack all $k+1$ copies simultaneously.
By an argument analogous to the above paragraphs, we are done.

\textbf{Set cover.} We have a ground set $E$, and a family $\mathcal{F}$ of sets $S_1, \dots, S_n \subseteq E$ over the ground set. 
We let $\U := \fromto{1}{n}$ and the goal is to pick a subset $I \subseteq \U$ of the indices such that $\bigcup_{i \in I} S_i = E$.
The attacker can attack up to $k$ of the indices $i \in I$ to forbid the corresponding sets from being picked.
We can make some index $i \in \U$ invulnerable, by simply duplicating the set $S_i$ a total amount of $k+1$ times.

Note that this satisfies the assumptions of \cref{thm:meta-theorem}, but modifies the family $\mathcal{F}$ such that the same set could appear multiple times in the family.
Alternatively, our construction can be adjusted such that this is avoided.
For this, we introduce $k+1$ new elements $e_1, \dots, e_{k+1}$ and $k+2$ new elements $f_1,\dots, f_{k+2}$ to the ground set $E$.
For each invulnerable index $i \in \fromto{1}{n} \setminus C$, we substitute $S_i$ by the $k+1$ sets $S_i^{(j)} = S_i \cup \{e_j\}$ for $j=1,\dots,k+1$.
Furthermore, we introduce $k+2$ new sets $S'_j := \fromto{e_1}{e_{k+1}} \cup \fromto{f_1}{f_{k+2}} \setminus \set{f_j}$ for $j =1, \dots, k+2$. This completes the description of the instance.
Note that the following holds: The elements $\fromto{f_1}{f_{k+2}}$ are covered by a set cover, if and only if it contains at least two sets of the form $S'_j$. 
Assuming this condition is true, all the elements $\fromto{e_1}{e_{k+1}}$ are already covered.
Hence all the different copies $S_i^{(j)}$ for $j=1,\dots,k+1$ are essentially equivalent.
Thus the attacker can not meaningfully attack all these copies simultaneously.
Note that the attacker can also not meaningfully attack the sets $S'_j$, since no matter which $k$ of them are attacked, 2 of them always remain.

\textbf{Steiner tree.} We have $\U = E$ in this case.
To make an edge $uv \in E \setminus C$ invulnerable, we replace it with $k+1$ parallel subdivided edges, i.e.\ we introduce vertices $w_1, \dots, w_{k+1}$ and edges $uw_i$ and $w_iv$ for $i =1,\dots, k+1$.
Every vulnerable edge $uv$ is replaced with only a single subdivided edge, i.e.\ a vertex $w$ and edges $uw, wv$.
It is clear that the number of edges of a minimum Steiner tree in the new instance is exactly two times as big as before, and the edge $uv$ has become effectively invulnerable.  

\textbf{Two vertex-disjoint path.}
We have $\U = A$ in this case.
The gadget is the same as for Steiner tree, except that the construction is directed, i.e. the arc $(u,v)$ is replaced either by the arcs $(u,w_i), (w_i, v)$ for $i=1,\dots,k+1$ (invulnerable case) or by the two arcs $(u,w), (w,v)$ (vulnerable case).
Since the paths in this problem have to be vertex disjoint, adding additional subdivided arcs between two existing vertices does not produce additional solutions because traveling from $u$ to $v$ renders all other paths from $u$ to $v$ unusable.

\textbf{Feedback arc set.} We have $\U = A$ in this case. 
Note that making some arc $a = (u,v) \in A \setminus C$ invulnerable means to ensure that it can be used in a minimal feedback arc set, no matter which $k$ arcs the attacker chooses.
This can be achieved the following way: Subdivide $a$ into $k + 1$ arcs. 
Clearly, the set of cycles in the new graph stays essentially the same. 
Furthermore, the attacker cannot block all $k + 1$ arcs from being chosen for the solution.
Choosing one of the subdivided pieces of $a$ in the new instance has the same effect as choosing $e$ in the old instance.

\textbf{Feedback vertex set.}
We have $\U = V$ in this case.
To make a vertex $v \in V \setminus C$ invulnerable, we split it into two vertices $v_\text{in}$ and $v_\text{out}$, 
put all incoming edges of the old vertex $v$ to $v_\text{in}$, 
put all outgoing edges of the old vertex $v$ to $v_\text{out}$,
and connect $v_\text{in}$ to $v_\text{out}$ with a directed path $P_v$ on $k+1$ vertices.
Note that in the new instance, a directed cycle uses one vertex of $P_v$ if and only if the cycle uses all vertices of $P_v$ if and only if a corresponding cycle in the old instance uses $v$.
By an analogous to argument to the feedback arc set case, we are done.

\textbf{Uncapacitated facility location.} We have $\U = J$ in this case, where $J$ is the set of sites for potential facilities. The attacker selects facility sites and forbids the decision maker to build a facility there. 
To make a facility site $j \in J \setminus C$ invulnerable, we can simply delete the site and replace it with $k+1$ identical sites, i.e.\ sites which have the same facility opening cost and service cost functions as the original facility $j$. 
Clearly, this way the attacker can not stop one of the equivalent facilities to be opened. On the other hand, since the facilities are identical (and uncapacitated), 
the decision maker has no advantage from opening two identical copies of the same facility.
Hence the new instance is identical to the old instance, with the only difference that facility site $j$ is invulnerable.

\textbf{$p$-median, $p$-center.} The difference between the facility location problem and the $p$-center and $p$-median problem is that in the latter two, there are no facility opening costs, at most $p$ facilities are allowed to be opened, 
and the service costs in the $p$-center problem are calculated using a minimum, and in the $p$-median problem they are calculated using the sum. 
All of these differences do not affect the argument from above, i.e.\ one can still make a facility site invulnerable by creating $k+1$ identical facilities. Hence the same argument holds.

\textbf{Subset Sum.}
We have $\U = \fromto{1}{n}$ and are given numbers $a_1, \dots, a_n \in \N$ and a target value $T$. The question is whether there exists $S \subseteq U$ with $\sum_{i \in S} a_i = T$. 
Consider some index $i \in \U \setminus C$. In order to make the index $i$ invulnerable, the first idea is to copy the number $a_i$ a total amount of $k+1$ times. 
But there is a problem with this construction -- if we do this, then the same number $a_i$ could be picked multiple times, which is not allowed in the original instance.
We need an additional gadget to make sure that $a_i$ gets used at most once for each $i$. This can be done the following way: 
The new instance contains the following numbers: Choose some number $B > 2k+2$ as a basis. 
For each $i \in C$, it contains the single number $B^{n(k+1)}a_i$. 
For each $i \in \fromto{1}{n} \setminus C$, it contains the $k+1$ distinct numbers  $c_i^{(j)} := B^{n(k+1)}a_i + B^{(i-1)(k+1) + j}$ for $j = 0,\dots, k$ as well as the $k+1$ distinct numbers $d_i^{(j)} := \sum_{\ell = 0,\ell \neq j}^k B^{(i-1)(k+1) + \ell}$ for $j = 0,\dots, k$ and the $k+1$ distinct numbers $e_i^{(j)} := B^{(i-1)(k+1) + j}$ for $j = 0,\dots, k$. We call $d_i^{(j)}$ and $e_i^{(j)}$ the helper numbers.
The new instance contains a total of $|C| + 3(k+1)(n - |C|)$ numbers. The new target value is 
$$
T' := B^{n(k+1)}T \ + \sum_{i \in \fromto{1}{n} \setminus C} \quad \sum_{\ell = 0}^k B^{(i-1)(k+1) + \ell}.
$$
Note that this has the following effect: 
Consider the representation of all involved numbers in base $B$. Let us call the digits $0$ up to $n(k+1) - 1$ the lower positions. 
Note that in the lower positions there can never be any carry, since for every lower position, all involved numbers have either a zero or one in that position and less than $B$ numbers have a one in the same place.
Due to that fact, in the lower positions the target $T'$ is reached if and only if for every $i \in \fromto{1}{n} \setminus C$, the corresponding \enquote{bitmask} is filled out (by this, we mean the positions $(i-1)(k+1)$ up to $i(k+1) - 1$).
This is achieved if and only if for some $j \in \fromto{0}{k}$ both the values $c^{(j)}_i$ and $d^{(j)}_i$ or both the values $d^{(j)}_i$ and $e^{(j)}_i$ are picked. In particular, at most one of the $k+1$ values $c^{(j)}_i$ for $j=0,\dots,k$ are picked.
In the upper positions, the target $T'$ is reached if and only if the corresponding choice in the old instance meets the target $T$.

Consider an attack of $k+1$ numbers by the attacker. For each $i \in \fromto{1}{n} \setminus C$ it holds that there exists a $j$ such that both $c_i^{(j)}$ and $d_i^{(j)}$ are not attacked. Likewise there exists a $j'$ such that both $d_i^{(j')}$ and $e_i^{(j')}$ are not attacked. 
That means that if $i$ is an invulnerable index, then no matter which $k+1$ values of  $c_i^{(j)}$, $d_i^{(j)}$ and  $e_i^{(j)}$ are attacked, 
a correct solution of subset sum will take for some $j$ either both $c_i^{(j)}$ and $d_i^{(j)}$ (which corresponds to taking $a_i$ in the original instance) 
or take both $d_i^{(j)}$ and $e_i^{(j)}$ (which corresponds to not taking $a_i$ in the original instance).
It follows that it is possible to block the new instance by attacking $k+1$ values if and only if it is possible to block the old instance by attacking $k+1$ of the vulnerable values. 
This was to show.
Finally, if the old numbers $a_1, \dots, a_n$ are pairwise distinct, the new numbers are as well. Hence the interdiction problem for subset sum is $\Sigma^p_2$-complete, even if all involved numbers are distinct.

\textbf{Knapsack.} The knapsack problem can be seen as a more general version of the subset sum problem, by creating for each $i$ 
from the subset sum instance a knapsack item with both profit $p_i = a_i$ and weight $w_i = a_i$, and setting both the weight and profit threshold to $T$.
Hence the $\Sigma^p_2$-completeness of \textsc{Min Cardinality Interdiction-Knapsack} follows as a consequence of the $\Sigma^p_2$-completeness of \textsc{Min Cardinality Interdiction-Subset Sum}. This holds even if all the involved knapsack items are distinct.

\subsection{An Invulnerability Reduction for Hamiltonian Cycle}
The invulnerability gadget for Hamiltonian cycle is the most involved of all our constructions, 
hence we devote a subsection to it. 
The main result in this section is that the minimum cardinality interdiction problem is $\Sigma^p_2$-complete for the nominal problems of both directed and undirected Hamiltonian cycle and path, as well as the TSP.

We present our reduction for the case of undirected Hamiltonian cycle and then argue how it can be adapted to the other cases. The main idea is to consider as an intermediate step only 3-regular graphs $G = (V, E)$, and then for a subset $C \subseteq E$ show how $E \setminus C$ can be made invulnerable. To this end, consider the SSP problem

\begin{description}
    \item[]\textsc{3Reg Ham}\hfill\\
    \textbf{Instances:} Undirected, 3-regular Graph $G = (V, E)$\\
    \textbf{Universe:} $\U := E$.\\
    \textbf{Solution set:} The set of all Hamiltonian cycles in $G$.
\end{description}

Recall that it is shown in \cite{gruene2024completeness} that \textsc{Hamiltonian Cycle} is SSP-NP-complete. We now require the stronger statement

\begin{lemma}
\label{lem:3-reg-ham-ssp}
    \textsc{3Reg Ham} is SSP-NP-complete.
\end{lemma}
\begin{proof}
    Garey, Johnson \& Tarjan \cite{DBLP:journals/siamcomp/GareyJT76} give a reduction from \textsc{3Sat} to  \textsc{3Reg Ham}, such that for every variable $x_i$ in the \textsc{3Sat} instance the graph $G$ has two distinct edges $e(x_i)$ and $e(\overline x_i)$ (compare Figure 7 in \cite{DBLP:journals/siamcomp/GareyJT76}). Let $E' := \bigcup_i \set{e(x_i), e(\overline x_i)}$ be the set of all these edges. For some assignment $\alpha$ of the \textsc{3Sat} variables, we say that $\alpha$ corresponds to the edge set $E_\alpha$ defined by $\set{e(x_i) : \alpha(x_i) = 1} \cup \set{e(\overline x_i) : \alpha(x_i) = 0}$.
    Garey, Johnson \& Tarjan show that there is a bijection between satisfying assignments and edge sets $E'' \subseteq E'$ that can be subset of a Hamiltonian cycle. More formally: 1.) For every satisfying assignment $\alpha$, 
    if one considers the set $E_\alpha \subseteq E'$ of edges corresponding to that assignment, 
    there exists a Hamiltonian cycle $H$ extending $E_\alpha$, i.e.\ $H \cap E' = E_\alpha$. 
    2.) For every Hamiltonian cycle $H$, we have that $H \cap E'$ equals $E_\alpha$ for some satisfying assignment $\alpha$.
    In total, 1.) and 2.) together show that the reduction in \cite{DBLP:journals/siamcomp/GareyJT76} is an SSP-reduction. (By defining $f(x_i) := e(x_i), f(\overline x_i) := e(\overline x_i)$.)
\end{proof}

We remark that it follows from \cite{akiyama1980np,DBLP:journals/siamcomp/GareyJT76} by the same argument that the problem is even SSP-NP-complete if restricted to 3-regular, bipartite, planar, 2-connected graphs. However, for our arguments it suffices to consider 3-regular graphs.

Consider now an instance of \textsc{3Reg Ham}, i.e. a 3-regular undirected graph  $G = (V, E)$. Let $C \subseteq E$ be a subset of the edges and $k \in \N_0$ the attacker's budget. We call $C$ the vulnerable edges. Let $D := E(G) \setminus C$.
In the remainder of this section we describe and prove a construction how to make the edges in $D$ invulnerable.
We quickly sketch the main idea: To make an edge $e = ab$ invulnerable, we enlarge it by replacing it with a large clique $W'_{ab}$ making sure that $e$ can be traversed no matter which $k$ edges inside $W'_{ab}$  are attacked. 
We also blow up each vertex $a$ of the original graph into a clique $W_a$.
However, this introduces new vertices into the instance, and we need to make sure that a Hamiltonian cycle can always trivially visit all the new vertices.
At the same time however, it should still hold that a Hamiltonian cycle in the new graph should be able to enter and exit these new objects $W_a$ and $W'_{ab}$ at most once, since otherwise a corresponding cycle in the old graph $G$ would visit edges or vertices twice, which is of course forbidden.
We achieve this by associating to each edge $e = ab$ a star of edges $F_{ab}$ and argue that a Hamiltonian cycle can use (essentially) at most one edge of each star $F_{ab}$. 
Furthermore, we will show that the fact that $G$ is 3-regular implies that each clique $W_a$ can be traversed (essentially) only once.

We are ready to begin with the construction.
First, let the directed graph $\overrightarrow{G}$ result from $G$ by orienting its edges arbitrarily and $k$ be the budget of the attacker.
We construct an undirected graph $G' = (V', E')$ from $\overrightarrow{G}$ as follows: 
Let $n := |V(G)|$.
For each vertex $a \in V(\overrightarrow{G})$, let $d_a$ be the out-degree of $a$, and let $W_a$ be a set of $2d_a + 4k + 1$ vertices.
For each invulnerable edge $ab \in D$ in the old graph, let  $W'_{ab}$ be a set of $4k$ vertices.
The vertex set $V(G')$ of the new graph $G'$ is then defined by
\[ V(G') = \bigcup_{a \in V} W_a \cup \bigcup_{ab \in D} W'_{ab}.\]
\begin{figure}
    \centering
    \includegraphics[scale=1.0]{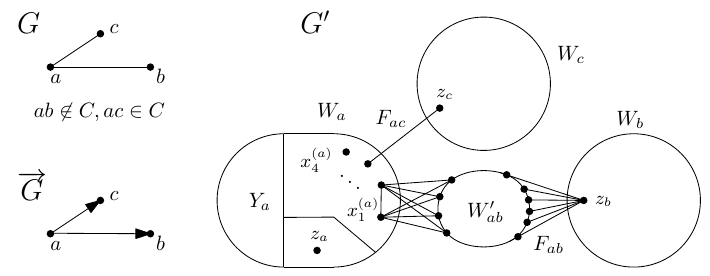}
    \caption{Invulnerability gagdet for Hamiltonian cycle which makes the edge $ab$ invulnerable while the edge $ac$ remains vulnerable.}
    \label{fig:ham-cycle-invulnerability}
\end{figure}
We further partition $W_a$ into three disjoint parts $W_a = X_a \cup Y_a \cup \set{z_a}$ of size $|X_a| = 2d_a$ and $|Y_a| = 4k$ and $|\set{z_a}| = 1$.
We denote the vertices of $X_a$ by $x^{(a)}_1, \dots, x^{(a)}_{2d_a}$.
The edges of $G'$ are defined as follows:
First, we let $W_a$ be a clique for all $v \in V$.
Second, for each vertex $a \in V$ in $\overrightarrow{G}$, let $e_1, \dots, e_{d_a}$ be its outgoing edges.
For each $i = 1, \dots, d_a$, consider the $i$-th outgoing edge $e_i = (a, b)$ of $a$, where $b$ is the corresponding neighbor. 
If $e_i \in C$, i.e. $e_i$ is vulnerable, then $G'$ contains simply the single edge $x^{(a)}_{2i-1}z_b$. 
In the other case, i.e.\ $e_i \in D$ is invulnerable, then $G'$ contains an invulnerability gadget as depicted in \cref{fig:ham-cycle-invulnerability} induced on the vertices $\set{x^{(a)}_{2i-1}, x^{(a)}_{2i}} \cup W'_{ab} \cup \set{z_b}$.
The invulnerability gadget consists out of a clique on the vertex set $\set{x^{(a)}_{2i-1}, x^{(a)}_{2i}} \cup W'_{ab}$, together with all edges from the set $W'_{ab}$ to the vertex $z_b$, i.e.\ a star centered at $z_b$ that has $W'_{ab}$ as its leaves.
Let $F_{ab}$ denote this star. 
Finally, for all vulnerable edges $ab \in D$, we also define $F_{ab}$ to be the single edge $x^{(a)}_{2i-1}z_b$ that connects $W_a$ to $W_b$.
This can be interpreted as a trivial star centered at $z_b$ with only one leaf.
This completes the description of $G'$.

The overall idea of this construction is that the cliques of $W_a$ cannot be attacked because they have at least $k$ vertices.
Thus it is always possible to find a path visiting all vertices of $W_a$.
Additionally, a star $F_{ab}$ of size larger than $k$ makes the edge $ab \in E$ invulnerable because at most $k$ many of the edges can be attacked.
Thus there is always the possibility to travel over one edge of $F_{ab}$ which corresponds to using edge $ab$ in the original graph.
On the other hand, since every edge of the star is connected to the same vertex $z_b$, we have that the star $F_{ab}$ can be used (essentially) exactly once.
Thus only the stars of size one (which correspond to the vulnerable edges) are attackable.
We now have everything that we need to prove our main result of this section. 

\begin{theorem}
\label{thm:ham-cycle-interdiction}
Minimum cardinality interdiction for \textsc{Undirected Hamiltonian Cycle} is $\Sigma^p_2$-complete.
\end{theorem}
\begin{proof}
Due to \cite{gruene2024completeness}, and \cref{lem:3-reg-ham-ssp}, we have that \textsc{Comb. Interdiction-3Reg Ham} is $\Sigma^p_2$-complete.
We claim that the construction of $G'$ yields a correct reduction from \textsc{Comb. Interdiction-3Reg Ham} to \textsc{Min Cardinality Interdiction-HamCycle}.
Indeed, the following two \cref{lem:hamcycle-if,lem:hamcycle-only-if} show that yes-instances of one problem get transformed into yes-instances of the other problem. 
\end{proof}
We remark that the 3-regularity of the graph is not maintained by the reduction.
(Indeed, an argument similar to the arguments given later in \cref{sec:noMeta} shows that the interdiction problem for Hamiltonian cycle restriced to only 3-regular graphs is likely not $\Sigma^p_2$-complete).

\begin{lemma}
\label{lem:hamcycle-if}
    If there exists $B' \subseteq E'$ of size $|B'| \leq k$, such that $G' - B'$ has no Hamiltonian cycle, then there is $B \subseteq C$ of size $|B| \leq k$ such that $G - B$ has no Hamiltonian cycle.
\end{lemma}
\begin{proof}
    Proof by contraposition. Assume that for all $B \subseteq C$ with $|B| \leq k$ the graph $G - B$ has a Hamiltonian cycle $H$.
    Given some $B' \subseteq E'$ with $|B'| \leq k$, we have to show that the graph $G' - B'$ has a Hamiltonian cycle. Let $B \subseteq C$ be the set of vulnerable edges in $G$ whose copies in $G'$ are attacked by $B'$ (i.e. $B = \set{ab \in C : F_{ab} \in B'}$). 
    Since $B \subseteq C$ and $|B| \leq k$, by assumption $G - B$ has a Hamiltonian cycle $H$. We want to modify $H$ to a Hamiltonian cycle of $H'$ of $G' - B'$. 
    The basic idea is to follow globally the same route as $H$. However, we have to pay attention, because we are not allowed to use edges from $B'$.
    For each vertex in $G'$ call it \emph{attacked}, if at least one of its incident edges are attacked by $B'$, and call it \emph{free} otherwise. 
    Note that since $|B'| \leq k$ and $|Y_a| = 4k$ and $|W'_{ab}| = 4k$ for $a \in V, ab \in E$, the vertex sets $Y_a$ and $W'_{ab}$ have at least $2k$ free vertices. Free vertices are good for the following reason: 
    Whenever we plan to go from some vertex $u$ to $v$ in $G'$, but we cannot because $uv \in B'$ was attacked, then we can instead choose any free vertex $f$ and go the route $u,f,v$ instead.
    Now the plan is that $H'$ will roughly employ the following strategy: Follow globally the same path in $G'$ like $H$ does in $G$. 
    Whenever $H'$ enters some new set $W_a$ for the first time, then we visit all the sets $W'_{ab}$ for all out-neighbors $b$ of $a$ in $\overrightarrow{G}$.
    Note that for such $b$, the set $W'_{ab}$ has two adjacent vertices with $W_a$ (we use these two vertices to enter and leave), and we collect all the vertices of $W'_{ab}$. 
    Here, we prioritize to visit first the attacked vertices of $W'_{ab}$ and then the remaining vertices of $W'_{ab}$. 
    After that, we collect all remaining vertices of $W_a$ (again prioritizing the attacked vertices first) before leaving $W_a$. (If the path on which we are leaving $W_a$ corresponds to an invulnerable edge $ab$ in $G$, we also collect all of $W'_{ab}$ in the process of leaving $W_a$.)

    Note that this plan might at first not be feasible, because it requires going over some edge $e' \in B'$. However note that, since $H$ does not use any edge of $B$, for every such edge $e'$ there are always at least $2k$ free vertices that are adjacent to both endpoints of $e'$.
    Hence it is possible to \enquote{repair} such an edge $e'$ by rerouting over some free vertex instead (and later skip over this free vertex). 
    Since there are at most $k$ defects, and there are at least $2k$ free vertices available at the end of traversing every set $W_a$ or $W'_{ab}$, all defects can be repaired. 
    Hence we can modify $H'$ to be a Hamiltonian cycle of $G' - B'$, which was to show.
\end{proof}

\begin{lemma}
\label{lem:hamcycle-only-if}
    If there exists $B \subseteq C$ of size $|B| \leq k$, such that $G - B$ has no Hamiltonian cycle, then there is $B' \subseteq E'$ of size $|B'| \leq k$ such that $G' - B'$ has no Hamiltonian cycle.
\end{lemma}
\begin{proof}
    Proof by contraposition. Assume that for all $B' \subseteq E'$ of size $|B'| \leq k$ the graph $G' - B'$ has a Hamiltonian cycle. 
    Given some $B \subseteq C$ with $|B| \leq k$, we have to show that the graph $G - B$ has a Hamiltonian cycle. 
    Let $B'$ be the trivial stars in $G'$ corresponding to the edges in $B$ (i.e.\ $B' = \set{F_{ab} : ab \in B}$). 
    Since $|B'| \leq k$, by assumption there is a Hamiltonian cycle $H'$ in $G' - B'$.  
    Consider the set $F := \bigcup_{ab \in E}E(F_{ab})$, i.e.\ the union of the edge sets of all the stars, trivial or not.
    We claim that w.l.o.g.\ we can assume that $|H' \cap F_{ab}| \leq 1$ for all $ab \in E$.
    Indeed, the graph $G' - F$ consists out of multiple connected components. 
    Each of these components contains exactly one set of the form $W_a$, and is incident to exactly three sets of the form $F_e$ in $G'$ (where $e$ is an edge that is either incoming to or outgoing from $a$ in $\overrightarrow{G}$).
    Suppose for some $F_{ab}$ we have $|H' \cap F_{ab}| \geq 2$.
    Since $F_{ab}$ is a star connected to a single vertex $z_b$, we have $|H' \cap F_{ab}| = 2$.
    Consider the edge $ab$ such that $F_{ab}$ connects the vertex $z_b$ with $W'_{ab}$. 
    By the observation about $G' - F$, the following is true about $H'$: 
    It enters $W'_{ab}$ in one of the two vertices attached to $X_a$, then traverses exactly all of $W'_{ab} \cup \set{z_b}$, 
    then leaves through the other of the two vertices attached to $X_a$, and at a later point returns to collect all vertices of $X_b \setminus \set{z_b}$.
    However, by the same observation as in \cref{lem:hamcycle-if}, if we define a free vertex to be a vertex not adjacent to any edge in $B'$, then both $W'_{ab}$ and $W_b$ have $2k$ free vertices.
    Hence we can modify $H'$ such that $H' \cap F_{ab} = \emptyset$.
    We thus assume that $|H' \cap F_{ab}| \leq 1$ for all $ab \in E$.
    Consider again the graph $G' - F$.
    Since each of its component is adjacent to three sets $F_e$ and $|H' \cap F_e| \leq 1$, we conclude that $H'$ uses exactly two of these three sets $F_e$.
    But this implies that $H'$ enters and exits each of the components of $G'- F$ only once and collects all of its vertices in the process.
    This implies that $H'$ globally follows the same path as some Hamiltonian cycle $H$ of $G$.
    Since $H' \subseteq G' - B'$, we conclude $H \subseteq G - B$.
    This was to show.
\end{proof}

These two lemmas together prove \cref{thm:ham-cycle-interdiction}.
We would now like to prove $\Sigma^p_2$-completeness also for Hamiltonian cycle interdiction of directed graphs.
Note that this does not follow from a trivial argument: 
Even though one can transform an undirected graph into a directed one, by substituting every undirected edge $uv$ by two directed edges $(u,v), (v,u)$, there is a problem: In the new setting the interdictor needs two attacks to separate $u,v$, while in the old setting the attacker only needs one. 

Still, the above proof can be adapted to the case of directed Hamiltonian cycle the following way: 
We start with \cite{plesn1979np}, which provides a SSP reduction to prove that the Hamiltonian cycle problem is NP-complete even in directed graphs $G$ such that $\text{indegree}(v) + \text{outdegree}(v) \leq 3$ for every vertex $v$, and such that for all pairs $u,v$ of $G$ at most one of the two edges $(u,v)$ and $(v,u)$ is present. 
Given a directed graph $G$, we then repeat the same construction as before, 
with the difference that we can start directly with the directed graph $G$ instead of obtaining an orientation $\overrightarrow{G}$ first.
This way, we can obtain an undirected graph $G'$ in the same way as before. 
In a final step, we turn $G'$ into a directed graph by substituting every undirected edge $uv$ by a pair of two edges $(u,v), (v,u)$. We perform this substitution for every edge of $G'$ with the exception of the edges that are part of some star $F_{ab}$.
Instead, for each star $F_{ab}$, we orient the edges of $F_{ab}$ the same way as the original directed edge of $G$ between $a,b$.
It can be shown that all the arguments from the above construction still hold. 
Hence the minimum cardinality interdiction problem is $\Sigma^p_2$-complete also for directed graphs.

If one is interested in Hamiltonian paths instead of cycles, a similar modification is possible.
Inspecting the proof of \cite{DBLP:journals/siamcomp/GareyJT76} (of \cite{plesn1979np}, respectively) more closely, we find that in both constructions the graph $G$ contains some edge $e = uv$ (some edge $e = (u,v)$, respectively) such that every Hamiltonian cycle uses $e$. 
We can delete $e$ and identify the vertices $s,t$ with the endpoints of $e$.
Then a Hamiltonian $s$-$t$-path in the new graph corresponds to a Hamiltonian cycle in the old graph and vice versa.
Note that this does not increase the degree of the graph.
The rest of the proof proceeds in the same manner, both in the undirected and directed case.
Finally, the proof can also easily be adapted to the TSP by a standard reduction of undirected Hamiltonian cycle to the TSP 
(a graph $G$ is transformed into a TSP instance on the complete graph where the costs obey $c(uv) = 1$ if $ev \in E(G)$ and $c(uv) = n+1$ if $uv \not\in E(G)$).
In conclusion, we have proven that the minimum cardinality interdiction problem is $\Sigma^p_2$-complete for the directed/undirected Hamiltonian path/cycle problem and the TSP. 

\section{Cases where the meta-theorem does not apply}\label{sec:noMeta}

It would be nice to establish a meta-theorem providing $\Sigma^p_2$-completeness of the minimum cardinality interdiction version of all nominal problems, which are SSP-NP-complete, instead of only those problems that admit an additional function $g$ with properties as stated in \Cref{thm:meta-theorem}.
However, we show in this section that this is not possible.
More precisely, we provide a lemma that guarantees that the minimum cardinality version of a problem in SSP-NP is in coNP.
Therefore, under the usual complexity-theoretic assumption NP $\neq \Sigma^p_2$, the interdiction problem is not $\Sigma^p_2$-complete.

In order to provide an intuition under which circumstances a minimum cardinality interdiction problem resides in the class coNP, we examine the vertex cover problem.
In a vertex cover, every edge $uv$ needs to be covered by at least one of the two incident vertices $u$ and $v$.
This, however, gives the attacker the opportunity to attack both $u$ and $v$ such that the edge $uv$ can never be covered.
Therefore, an attacker budget of at least $2$ results in a clear Yes-instance.
On the other hand, if the attacker budget if at most $1$, we can provide a certificate for No-instances.
We can summarize this observation in the following lemma.

\begin{lemma}\label{lem:minCardInCoNP}
    Let $\Pi = (\I, \U, \sol)$ be an SSP problem.
    If in each instance $I \in \I$ there is a subset $U' \subseteq \U(I)$ of constant size, i.e. $|U'| = O(1)$, such that for $U' \cap S \neq \emptyset$ for all $S \in \sol(I)$, then \textsc{Min Cardinality Interdiction-$\Pi$} is contained in coNP.
\end{lemma}
\begin{proof}
    Let $k$ be the interdiction budget.
    If $|U'| \leq k$, then the interdictor is able to block the whole set $U'$.
    By definition of $U'$, there is no solution $S \in \sol(I)$ such that $U' \cap S \neq \emptyset$ and thus the interdictor has a winning strategy.
    If on the other hand $k < |U'| = O(1)$, then there is a polynomially sized certificate encoding a winning strategy of the defender, i.e. a certificate for a No-instance of the problem.
    For this, we first encode the ${|\U(I)| \choose k} = |\U(I)|^{O(1)}$ possible blockers $B' \subseteq \U(I)$ and then the solution $S \in \mathcal S(I)$ such that $S \cap B' \neq \emptyset$ for all $B' \subseteq \U(I)$.
    It is possible to efficiently verify the solution by checking whether there is a solution $S \in \mathcal S(I)$ such that $S \cap B' \neq \emptyset$ for all $B' \subseteq \U(I)$ holds because the nominal problem $\Pi$ is in NP.
    It follows that the problem lies in coNP.
\end{proof}

Consider the different variants of interdiction problems introduced in \cref{sec:different-variants-of-interdiction}.
Since they are more general, \cref{lem:minCardInCoNP} does not immediately imply that those variants are contained in coNP.
However, if for each instance the stronger condition $U' \cap F \neq \emptyset$ for all feasible solutions $F \in \F(I)$ and for some constant size set $U' \subseteq \U(I)$ holds, then the \emph{full decision variant of interdiction} and the \emph{most vital element problem} are contained in coNP.
Besides the containment in coNP, we can also derive the following corollary pinpointing the complexity of minimum cardinality interdiction problems whose nominal problem is in SSP-NP.

\begin{corollary}
    Let $\Pi = (\I, \U, \sol)$ be an SSP-NP-complete problem.
    If in each instance $I \in \I$ there is a subset $U' \subseteq \U(I)$ of constant size, i.e. $|U'| = O(1)$, such that for $U' \cap S \neq \emptyset$ for all $S \in \sol(I)$, then \textsc{Min Cardinality Interdiction-$\Pi$} is coNP-complete.
\end{corollary}
\begin{proof}
    There is a reduction by restriction:
    Setting the interdiction budget $k = 0$ results in the corresponding co-problem co$\Pi$ of the nominal problem $\Pi$.
\end{proof}

\subsection{Applying the Lemma to Various Problems}

In this section, we apply \cref{lem:minCardInCoNP} to the problems mentioned earlier in this paper.
Some of the problems are affected in their original general form, e.g. vertex cover or satisfiability, while for others the lemma can be applied on a restricted version such as independent set on graphs with bounded minimum degree.
For this, we shortly describe the problem and then give the argument on how the lemma is applicable.

\textbf{Vertex Cover.}
An instance of the vertex cover interdiction problem consists of a graph $G$ and numbers $t,k \in \N_0$.
The question is if the attacker can find a set $B \subseteq V(G)$ with $|B| \leq k$ such that $B \cap S \neq \emptyset$ for every vertex cover $S$ of size at most $t$.
Now, observe that if $k \geq 2$ (and the graph is non-empty), the attacker can easily find such a set $B$ by selecting two adjacent vertices.
Thus, \cref{lem:minCardInCoNP} applies by defining $U' = \{u, v\}$ for some edge $uv \in E(G)$.
Observe that this not only destroys the solutions $S \in \sol(I)$ but also all feasible solutions $F \in \F(I)$.
Thus the minimum cardinality interdiction version, the full decision variant of interdiction and the most vital elements problem of vertex cover are coNP-complete.

\textbf{Satisfiability.}
An instance of the satisfiability interdiction problem consists out of a formula in CNF over the variables $X = \fromto{x_1}{x_n}$, with the literal set as universe, i.e. $\U = X \cup \overline X$, and interdiction budget $k$.
A similar issue as in vertex cover interdiction arises here:
If $k \geq 2$, the interdictor can just choose for some $i \in \fromto{1}{n}$ to attack both literals $x_i, \overline x_i$.
Every satisfying assignment (of non-trivial instances) contains either $x_i$ or $\overline{x_i}$, hence this is a successful attack.
Thus, \cref{lem:minCardInCoNP} applies by defining $U' = \{x, \overline x\}$ for some literal pair $x, \overline x \in \U$.
Again, this also destroys all feasible solutions $F \in \F(I)$.
Thus the minimum cardinality interdiction version, the full decision variant of interdiction and the most vital elements problem of satisfiability are coNP-complete.

\textbf{Independent Set on graphs with bounded minimum degree.}
An instance of the independent set interdiction problem consists of a graph $G =(V,E)$ with universe $U = V$, a threshold $t$ and an interdiction budget $k$.
The question of the independent set  problem is if there is a set $I \subseteq V$ such that all vertices in $I$ do not share an edge.
We now take the vertex $d$ of bounded degree into consideration.
If the attacker attacks the closed neighborhood $N[d]$ of $d$, all (optimal) solutions $S \in \mathcal S$ can be interdicted and thus \Cref{lem:minCardInCoNP} is applicable.
Thus minimum cardinality interdiction independent set on graphs with bounded minimum degree is coNP-complete.
In contrast to the other problems, this statement is not true for general feasible solutions $F \in \mathcal F(I)$. Hence we do not obtain a result for the variants from \Cref{sec:different-variants-of-interdiction}.

\textbf{Dominating Set on graphs with bounded minimum degree.}
An instance of the dominating set interdiction problem consists of a graph $G =(V,E)$ with universe $U = V$, a threshold $t$ and an interdiction budget $k$.
The question of the dominating set problem is if there is a set $D \subseteq V$ of size at most $t$ such that $D$ dominates all vertices of vertex set $V$. In other words, the union of the neighborhoods of the vertices in $D$ is the vertex set $V$, i.e. $\bigcup_{v \in D} N[v] = V$.
Again we consider a vertex $d$ of bounded degree.
Then, we can define the set of constant size to be $U' = N[d]$.
All feasible solutions $F \in \mathcal F$ have to include some vertex from $U'$ (otherwise $d$ would not be dominated).
Thus \Cref{lem:minCardInCoNP} is applicable to dominating set.
Accordingly, the minimum cardinality interdiction version, the full decision variant of interdiction and the most vital elements problem of dominating set on graphs with bounded minimum degree are coNP-complete.

\textbf{Hitting Set with bounded minimum set size.}
An instance of hitting set interdiction consists of a ground set $\{1, \ldots, n\}$ and $m$ sets $S_j \subseteq \{1,\ldots,n\}$ as well as a threshold $t$ and an interdiction budget $k$.
The universe is defined by $\U = \fromto{1}{n}$.
The question of the hitting set problem is whether there is a hitting set $H \subseteq \fromto{1}{n}$ of size at most $t$ for the sets $S_j$, that is, $H \cap S_j \neq \emptyset$ for $1 \leq j \leq m$.
We can apply \Cref{lem:minCardInCoNP} by defining $U'$ to be the set of constant size $|S_c| = O(1)$.
Then, the attacker is able to block the entire set $S_c$ such that it is not hittable, which interdicts all feasible solutions $F \in \mathcal F$.
Therefore the minimum cardinality interdiction version, the full decision variant of interdiction and the most vital elements problem of hitting set with bounded minimum set size are coNP-complete.

\textbf{Set Cover with bounded minimum coverage.}
An instance of the set cover interdiction problem consists of sets $S_i \subseteq \{1, \ldots, m\}$ for $1 \leq i \leq n$, a threshold $t$ and the an interdiction budget $k$.
The universe is defined as the sets $S_i$, $1 \leq i \leq n$.
The question of the set cover problem is whether there is selection $S \subseteq \{S_1, \ldots, S_n\}$ of size at most $k$ such that $\bigcup_{s \in S} s = \{1, \ldots, m\}$.
If there is an element $e \in \{1, \ldots, m\}$ of bounded coverage, i.e. there is a constant number of $S_i$, $1 \leq i \leq n$, with $e \in S_i$, then the attacker can attack all of these sets $S_i$.
Thus, we can apply \Cref{lem:minCardInCoNP} by choosing $U' = \{S_i \mid e \in S_i\}$ and all feasible solutions $F \in \mathcal F$ are blockable.
Accordingly, the minimum cardinality interdiction version, the full decision variant of interdiction and the most vital elements problem of set cover with bounded minimum coverage are coNP-complete.

\textbf{Steiner Tree on graphs with bounded minimum degree of terminal vertices.}
An instance of the Steiner tree interdiction problem consists of a graph $G= (S \cup T, E)$ of Steiner vertices $S$ and terminals $T$, edge weights $c: E \rightarrow \mathbb N$, a threshold $t$ and a interdiction budget $k$.
The universe is the edge set $\U = E$.
The question of the Steiner tree problem is if there is a tree $E' \subseteq E$ of weight $c(E') \leq t$ such that all terminal vertices $T$ are connected by $E'$.
If there is a terminal vertex $d \in T$ of bounded degree, then all incident edges build up a set $U' = \{dv \in E\}$ on which we can apply \Cref{lem:minCardInCoNP}.
This blocks all feasible solutions $F \in \mathcal F$.
Therefore, the minimum cardinality interdiction version, the full decision variant of interdiction and the most vital elements problem of Steiner tree on graphs with bounded minimum degree of terminal vertices are coNP-complete.

\textbf{Two Vertex-Disjoint Path on graphs with bounded degree.}
An instance of the two vertex-disjoint path interdiction problem consists of a directed graph $G=(V,A)$, vertices $s_1, s_2, t_1, t_2 \in V$ and interdiction budget $k$.
The universe is the arc set $\U = A$.
The question of the two vertex-disjoint path is if there are two paths $P_1, P_2 \subseteq A$ such that $P_i$ starts at $s_i$ and ends at $t_i$ and both paths $P_1$ and $P_2$ do not share a vertex.
If the the graph has bounded degree, we can choose any of the vertices that have to be included in on of the paths, e.g. $s_1$, and include all the incident arcs in $U' = \{(s_1, v) \in A\}$ such that we can apply \Cref{lem:minCardInCoNP}.
This blocks all feasible solutions $F \in \mathcal F$.
Accordingly, the minimum cardinality interdiction version, the full decision variant of interdiction and the most vital elements problem of two vertex-disjoint path on graphs with bounded degree are coNP-complete.

\textbf{Feedback Vertex Set on graphs with bounded girth.}
An instance of the feedback vertex set interdiction problem consists of a directed graph $G=(V,A)$, a threshold $t$ and interdiction budget $k$.
The universe is the vertex set $\U = V$.
The question of feedback vertex set is if there is a set $V' \subseteq V$ such that the graph is cycle free.
Accordingly, if the graph has bounded girth, there is a cycle of bounded length, which the attacker can attack or in other words, the cycle cannot be deleted by the defender by choosing a corresponding vertex to be in the feedback vertex set.
Thus all feasible solutions $F \in \mathcal F$ are blockable by applying \Cref{lem:minCardInCoNP} with $U' = \{v \in V \mid v \text{ is part of the smallest cycle in } G\}$.
Therefore, the minimum cardinality interdiction version, the full decision variant of interdiction and the most vital elements problem of feedback vertex set on graphs with bounded girth are coNP-complete.

\textbf{Feedback Arc Set on graphs with bounded girth.}
An instance of the feedback arc set interdiction problem consists of a directed graph $G=(V,A)$, a threshold $t$ and interdiction budget $k$.
The universe is the arc set $\U = A$.
The question of feedback arc set is if there is an arc set $A' \subseteq A$ such that the graph is acyclic.
We can use the same argument as in feedback vertex set.
That is, the attacker can choose the arcs of the smallest cycle in $G$.
Thus all feasible solutions $F \in \mathcal F$ are blockable by applying \Cref{lem:minCardInCoNP} with $U' = \{a \in A \mid a \text{ is part of the smallest cycle in } G\}$.
Therefore, the minimum cardinality interdiction version, the full decision variant of interdiction and the most vital elements problem of feedback arc set on graphs with bounded girth are coNP-complete.

\textbf{Uncapacitated Facility Location, p-Center, p-Median with bounded minimum customer coverage.}
An instance of the minimum cardinality interdiction version of these three problems consists of a set of potential facilities $F$ and a set of clients $C$ together with a cost function on the facilities $f: F \rightarrow \mathbb N$ and a service cost function $c: F \times C \rightarrow \mathbb N$ as well as a threshold $t$ and an interdiction budget $k$.
The universe is the facility set $\U = F$ and it is asked for a set of facilities $F' \subseteq F$ not exceeding the cost threshold $t$.
If the coverage of one customer is bounded, i.e. there is a bounded number of potential facilities that are able to serve the customer, the attacker is able to block all of these.
Thus we can define $U'$ as the set of facilities that are able to serve the customer of bounded coverage such that all feasible solutions $F \in \mathcal F$ can be interdicted.
Therefore, we can apply \Cref{lem:minCardInCoNP} and the minimum cardinality interdiction version, the full decision variant of interdiction and the most vital elements problem of these three facility locations problems with bounded minimum customer coverage are coNP-complete.

\textbf{Hamiltonian path/cycle (directed/undirected), TSP on graphs with bounded minimum degree.}
An instance of the minimum cardinality interdiction version of these problems consists of a graph $G=(V,E)$ (respectively $G=(V,A)$ in the directed case) and an interdiction budget $k$.
The universe is the set of edges $\U = E$ (respectively the set of arcs $\U = A$).
The question is whether there is a Hamiltonian path or cycle in $G$, i.e. a path/cycle that visits every vertex exactly once.
Because there is a vertex $d$ of bounded degree which has to be visited, we can define the set of constant size $U' = \{dv \in E\}$ (respectively $U' = \{(d,v),(v,d) \in A\}$).
If the set $U'$ is blocked it is not possible to visit the vertex, thus all feasible solutions $F \in \mathcal F$ can be interdicted.
Therefore, we can apply \Cref{lem:minCardInCoNP} and the minimum cardinality interdiction version, the full decision variant of interdiction and the most vital elements problem of these five Hamiltonian problems on graphs with bounded minimum degree are coNP-complete.

\subsection{Satisfiability with Universe over the Variables}

In the previous subsection we explained why minimum cardinality interdiction-\textsc{Sat} is contained in coNP, hence likely not $\Sigma^p_2$-complete.
Note that this is a consequence of our choice of definition of \textsc{Satisfiability}, where we explicitly defined the universe to be the literal set $L = X \cup \overline X$.
As a consequence, the interdictor may attack $X \cup \overline X$. 
\begin{samepage}
    \begin{mdframed}
    	\begin{description}
        \item[]\textsc{Satisfiability ($\U = L$)}\hfill\\
        \textbf{Instances:} Literal Set $L = \fromto{x_1}{x_n} \cup \fromto{\overline x_1}{\overline x_n}$, Clauses $C \subseteq \powerset{L}$\\
        \textbf{Universe:} $L =: \U$.\\
        \textbf{Solution set:} The set of all sets $L' \subseteq \U$ such that for all $i \in \fromto{1}{n}$ we have $|L' \cap \set{\ell_i, \overline \ell_i}| = 1$, and such that $|L' \cap c_j| \geq 1$ for all $c_j \in C$.
    	\end{description}
    \end{mdframed}
\end{samepage}

An interesting behavior occurs, when we consider the following alternative version \textsc{Satisfiability ($\U = X$)}. 
\begin{samepage}
    \begin{mdframed}
    	\begin{description}
        \item[]\textsc{Satisfiability  ($\U = X$)}\hfill\\
        \textbf{Instances:} Variable Set $X = \fromto{x_1}{x_n}$, Clauses $C \subseteq 2^{X \cup \overline X}$ \\
        \textbf{Universe:} $X =: \U$.\\
        \textbf{Solution set:} The set of all sets $X' \subseteq \U$ such that the assignment $\alpha: X \rightarrow \{0,1\}$ with $\alpha(x) = 1 \leftrightarrow x \in X'$ satisfies all clauses in $C$.
        \end{description}
    \end{mdframed}
\end{samepage}
Here the universe is only the variable set $X$, so in the interdiction version, the interdictor may only attack $X$, i.e.\ the interdictor may target individual variables and enforce that they must be set to \emph{false}. 
We show now that in contrast to the variant, where the universe is the literal set, in this new variant the interdiction problem is $\Sigma^p_2$-complete again. 
Since the problem \textsc{Satisfiability ($\U = X$)} is not part of the original problem set of \cite{gruene2024completeness}, we perform this proof in two steps.
\begin{lemma}
    \textsc{Satisfiability ($\U = X$)} is SSP-NP-complete, even when all clauses are restricted to length at most three.
\end{lemma}
\begin{proof}
    We provide an SSP reduction from the SSP-NP-complete problem \textsc{Satisfiability ($\U = L$)} to \textsc{Satisfiability ($\U = X$)}. 
    Consider an instance of \textsc{Satisfiability ($\U = L$)} given by a formula $\varphi$ with $n$ variables $X = \fromto{x_1}{x_n}$ and universe/literal set $\U = L = X \cup \overline{X}$.
    \textsc{Satisfiability ($\U = L$)} is SSP-NP-complete even when all clauses are restricted to length three, so let us w.l.o.g.\ assume that property.
    We have to show how to embed this universe into the universe $\U'$ of some corresponding \textsc{Satisfiability ($\U = X$)} instance $\varphi'$, where only positive literals are allowed in $\U'$.
    This can be done the following way:
    We introduce $2n$ new variables $X' := \fromto{x^t_1}{x^t_n} \cup \fromto{x^f_1}{x^f_n}$.
    The universe $\U' := X'$ consists out of the $2n$ corresponding positive literals $X'$.
    The new formula $\varphi'$ is defined from $\varphi$ in two steps.
    First a substitution process takes place: 
    For each $i=1,\dots,n$, the positive literal $x_i$ is replaced by the positive literal $x^t_i$ and each negative literal $\overline x_i$ is replaced by the positive literal $x_i^f$.
    In a second step, the clauses $(x^t_i \lor \overline x^f_i) \land (\overline x^t_i \lor x^f_i)$ (note that these are equivalent to $x_i^t \oplus x_i^f$) are added to $\varphi'$.
    Formally,
    \[
        \varphi' = \text{substitute}(\varphi) \land \bigwedge_{i=1}^n (x_i^t \lor x_i^f)\land (\overline{x}_i^t \lor \overline{x}^f_i).
    \]
    The SSP reduction is completed by specifying the embedding function $f : \U \to \U'$ via $f(x_i) := x_i^t$ and $f(\overline x_i) := x_i^f$.
    Clearly all clauses of $\varphi'$ have length at most three.
    Note that this reduction is a correct reduction, i.e.\ it transforms yes-instances into yes-instances and no-instances into no-instances, because the added constraints make sure that exactly one of $x_i^t$ and $x_i^f$ is true.
    Furthermore, it has the SSP property:
    For every solution $S \subseteq \U$ of \textsc{Satisfiability ($\U = L$)}, the \enquote{translated} set $f(S) \subseteq \U'$ is a solution of \textsc{Satisfiability ($\U = X$)}.
    Furthermore, for every solution $S \subseteq \U'$ of \textsc{Satisfiability ($\U = X$)}, the set $f^{-1}(S)\subseteq \U$ is a solution of \textsc{Satisfiability ($\U = L$)}.
    Accordingly, we have a correct SSP reduction (where the SSP mapping $f$ is even bijective due to $f(\U) = \U'$). 
\end{proof}

\begin{theorem}
    \textsc{Min Cardinality Interdiction-Satisfiability ($\U = X$)} is $\Sigma^p_2$-complete.
\end{theorem}
\begin{proof}
    By the previous lemma, \textsc{Satisfiability ($\U = X$)} is SSP-NP-complete, even if all clauses are restricted to length three.
    Due to \cite{gruene2024completeness}, the problem \textsc{Comb. Interdiction-Satisfiability ($\U = X$)} is $\Sigma^p_2$-complete, even if all clauses are restricted to length three.
    We provide a reduction from the latter problem in terms of an invulnerability gadget analogous to the gadgets presented in \cref{sec:invulnerability-gadgets}. 
    For this, consider an instance of \textsc{Satisfiability ($\U = X$)} with formula $\varphi$ in CNF and every clause of length three, together with the universe $\U = \fromto{x_1}{x_n}$, a set $C \subseteq \U$ of vulnerable literals, and interdiction budget $k \in \N_0$.
    For every variable $x_i \in \U \setminus C$, we explain how to make $x_i$ invulnerable.
    We introduce $k+1$ new variables $x^{(1)}_i, \dots x^{(k+1)}_i$.
    Our goal is to establish the equivalence
    \[
        x_i \equiv x^{(1)}_i \lor \dots \lor x^{(k+1)}_i.
    \]
    We can achieve this through means of the following substitution process starting from formula $\varphi$: 
    Every occurrence of $x_i$ in the formula gets substituted by $x^{(1)}_i \lor \dots \lor x^{(k+1)}_i$. 
    Every occurrence of $\overline x_i$ gets substituted (by De Morgan's law) by $(\overline x^{(1)}_i \land \dots \land \overline x^{(k+1)}_i)$.
    Note that this has two effects: 
    First, the length of a clause may now exceed 3.
    Secondly, the formula is not in CNF anymore. 
    Note however that we can use the distributive law to expand every clause that is not in CNF. 
    Since before each clause before had a length of at most three, this results in a blow-up of the instance size of a factor at most $(k+1)^3$, i.e.\ at most a polynomial factor.
    Let $\varphi'$ be the resulting formula. 
    We can see that there is an equivalence of the satisfying assignments of $\varphi$ and $\varphi'$, in the sense that $x_i$ is true in $\varphi$ if and only if $x^{(1)}_i \lor \dots \lor x^{(k+1)}_i$ is true in $\varphi'$ (for all invulnerable $x_i$). 
    However, since the interdiction budget is only $k$, the interdictor can never enforce $x^{(1)}_i \lor \dots \lor x^{(k+1)}_i$ to be false for all invulnerable variables.
    This shows that \textsc{Comb. Interdiction-Satisfiability ($\U = X$)} reduces to \textsc{Min. Cardinality Interdiction-Satisfiability ($\U = X$)}, hence proving its $\Sigma^p_2$-completeness.
\end{proof}

Note that the reasoning presented in this proof was slightly different from \cref{thm:meta-theorem}, since we start with a formula where every clause has length three, but do not preserve this property during the proof.
Hence $\Sigma^p_2$-completeness is only shown in the case where clauses can have unrestricted length.

We can use an argument similar to \Cref{lem:minCardInCoNP} to show the coNP-completeness of the minimum cardinality interdiction version, the full decision variant of interdiction and the most vital elements problem of {\sc $b$-Satisfiability ($\U = X$)}, i.e. with clauses of length bounded by $b$.
Indeed, it is easy to see that the interdiction problem of \textsc{Satisfiability ($\U = X$)} where every clause has length three is coNP-complete:
If $k \geq 3$ holds for the interdiction budget, the attacker distinguishes two cases:
If there is a clause with three positive literals, the attacker blocks all of them and immediately wins. 
In the other case, every clause has at least one negative literal.
Then the attacker can never win, since the defender can set every variable to false, which is a satisfying assignment that can never be blocked.
By an analogous argument, we can see that for any $t = O(1)$, the interdiction problem of \textsc{Satisfiability ($\U = X$)} with clauses restricted to length $t$ is coNP-complete.

Finally, we remark that slightly different variants of interdiction-3-Sat have been shown to be $\Sigma^p_2$-complete. In these variants, the interdictor does not have access to all variables (see \cite[Sec. 4.2]{gruene2024completeness} or \cite[Thm. 1]{jackiewicz2024computational}).

\section{Conclusion}
We have shown that for a large class of NP-complete problems, the corresponding minimum cardinality interdiction problem is $\Sigma^2_p$-complete.
With that we have also shown the hardness of several different variants of interdiction that can be found in the literature including minimum blocker and most vital elements problems.
For this, we introduced a new type of reduction, namely invulnerability reductions.
This reduction uses the corresponding minimum cost interdiction problem as basis and ensures that non-blockable elements are effectively not attackable.
The hardness of the minimum cost interdiction problem is provable via an SSP reduction.
Additionally, we show that for some problems (e.g. vertex cover, satisfiability), the $\Sigma^p_2$-completeness cannot be derived despite the fact that the minimum cost interdiction problem is $\Sigma^p_2$-complete.
Overall, we show for 23 minimum cardinality interdiction problems that they are either $\Sigma^p_2$-complete or coNP-complete.
with the ability to apply the framework to further problems.

The following natural questions arise.
First, it is of interest to find more problems for which this framework is applicable.
Furthermore, it is relevant whether this framework is also extendable to problems that are in NP but not NP-complete.
One might lose the $\Sigma^2_p$-completeness for these problems, however, a meta-theorem that proves NP-completeness for such problems and generalizes the existing results in the literature is important to obtain a deeper understanding on the structure of such problems.
At last, the results of this paper are not always applicable (albeit sometimes) to the vertex deletion or in general element deletion interdiction problem.
Thus, it is of interest to show a similar meta-theorem for element deletion problems.

\newpage

\bibliography{bib_general,bib_interdiction,bib_reductions}

\newpage

\appendix
\section{Problems Definitions}
\label{app:sec:problemDefinitions}

\begin{samepage}
    \begin{mdframed}
    	\begin{description}
        \item[]\textsc{Satisfiability}\hfill\\
        \textbf{Instances:} Literal Set $L = \fromto{\ell_1}{\ell_n} \cup \fromto{\overline \ell_1}{\overline \ell_n}$, Clauses $C \subseteq \powerset{L}$.\\
        \textbf{Universe:} $L =: \U$.\\
        \textbf{Solution set:} The set of all sets $L' \subseteq \U$ such that for all $i \in \fromto{1}{n}$ we have $|L' \cap \set{\ell_i, \overline \ell_i}| = 1$, and such that $|L' \cap c_j| \geq 1$ for all $c_j \in C$, $j \in \fromto{1}{|C|}$.
    	\end{description}
    \end{mdframed}
\end{samepage}

\begin{samepage}
    \begin{mdframed}
    	\begin{description}
        \item[]\textsc{3-Satisfiability}\hfill\\
        \textbf{Instances:} Literal Set $L = \fromto{\ell_1}{\ell_n} \cup \fromto{\overline \ell_1}{\overline \ell_n}$, Clauses $C \subseteq \powerset{L}$ s.t. $\forall c_j \in C : |c_j| = 3$.\\
        \textbf{Universe:} $L =: \U$.\\
        \textbf{Solution set:} The set of all sets $L' \subseteq \U$ such that for all $i \in \fromto{1}{n}$ we have $|L' \cap \set{\ell_i, \overline \ell_i}| = 1$, and such that $|L' \cap c_j| \geq 1$ for all $c_j \in C$.
    	\end{description}
    \end{mdframed}
\end{samepage}

\begin{samepage}
    \begin{mdframed}
    	\begin{description}
        \item[]\textsc{Dominating Set}\hfill\\
        \textbf{Instances:} Graph $G = (V, E)$, number $k \in \N$.\\
        \textbf{Universe:} Vertex set $V =: \U$.\\
        \textbf{Feasible solution set:} The set of all dominating sets.\\
        \textbf{Solution set:} The set of all dominating sets of size at most $k$.
    	\end{description}
    \end{mdframed}
\end{samepage}

\begin{samepage}
    \begin{mdframed}
    	\begin{description}
        \item[]\textsc{Set Cover}\hfill\\
        \textbf{Instances:} Sets $S_i \subseteq \fromto{1}{m}$ for $i \in \fromto{1}{n}$, number $k \in \N$.\\
        \textbf{Universe:} $\{S_1 \dots, S_n\} =: \U$.\\
        \textbf{Feasible solution set:} The set of all $S \subseteq \{S_1, \dots, S_n\}$ s.t. $\bigcup_{s \in S} s = \fromto{1}{m}$.\\
        \textbf{Solution set:} Set of all feasible solutions with $|S| \leq k$.
    	\end{description}
    \end{mdframed}
\end{samepage}

\begin{samepage}
    \begin{mdframed}
    	\begin{description}   
        \item[]\textsc{Hitting Set}\hfill\\
        \textbf{Instances:} Sets $S_j \subseteq \fromto{1}{n}$ for $j \in \fromto{1}{m}$, number $k \in \N$.\\
        \textbf{Universe:} $\fromto{1}{n} =: \U$.\\
        \textbf{Feasible solution set:} The set of all $H \subseteq \fromto{1}{n}$ such that $H \cap S_j \neq \emptyset$ for all $j \in \fromto{1}{m}$.\\
        \textbf{Solution set:} Set of all feasible solutions with $|H| \leq k$.
    	\end{description}
    \end{mdframed}
\end{samepage}

\begin{samepage}
    \begin{mdframed}
    	\begin{description}
        \item[]\textsc{Feedback Vertex Set}\hfill\\
        \textbf{Instances:} Directed Graph $G = (V, A)$, number $k \in \N$.\\
        \textbf{Universe:} Vertex set $V =: \U$.\\
        \textbf{Feasible solution set:} The set of all vertex sets $V' \subseteq V$ such that after deleting $V'$ from $G$, the resulting graph is cycle-free (i.e. a forest).\\
        \textbf{Solution set:} The set of all feasible solutions $V'$ of size at most $k$.
    	\end{description}
    \end{mdframed}
\end{samepage}

\begin{samepage}
    \begin{mdframed}
    	\begin{description}
        \item[]\textsc{Feedback Arc Set}\hfill\\
        \textbf{Instances:} Directed Graph $G = (V, A)$, number $k \in \N$.\\
        \textbf{Universe:} Arc set $A =: \U$.\\
        \textbf{Feasible solution set:} The set of all arc sets $A' \subseteq A$ such that after deleting $A'$ from $G$, the resulting graph is cycle-free (i.e. a forest).\\
        \textbf{Solution set:} The set of all feasible solutions $A'$ of size at most $k$.
    	\end{description}
    \end{mdframed}
\end{samepage}

\begin{samepage}
    \begin{mdframed}
    	\begin{description} 
        \item[]\textsc{Uncapacitated Facility Location}\hfill\\
        \textbf{Instances:} Set of potential facilities $F = \fromto{1}{n}$, set of clients $C = \fromto{1}{m}$, fixed cost of opening facility function $f: F \rightarrow \Z$, service cost function $c: F \times C \rightarrow \Z$, cost threshold $k \in \Z$\\
        \textbf{Universe:} Facility set $F =: \U$.\\
        \textbf{Solution set:} The set of sets $F' \subseteq F$ s.t. $\sum_{i \in F'} f(i) + \sum_{j \in C} \min_{i \in F'} c(i, j) \leq k$.
    	\end{description}
    \end{mdframed}
\end{samepage}

\begin{samepage}
    \begin{mdframed}
    	\begin{description} 
        \item[]\textsc{p-Center}\hfill\\
        \textbf{Instances:} Set of potential facilities $F = \fromto{1}{n}$, set of clients $C = \fromto{1}{m}$, service cost function $c: F \times C \rightarrow \Z$, facility threshold $p \in \N$, cost threshold $k \in \Z$\\
        \textbf{Universe:} Facility set $F =: \U$.\\
        \textbf{Solution set:} The set of sets $F' \subseteq F$ s.t. $|F'| \leq p$ and $\max_{j \in C} \min_{i \in F'} c(i, j) \leq k$.
    	\end{description}
    \end{mdframed}
\end{samepage}

\begin{samepage}
    \begin{mdframed}
    	\begin{description} 
        \item[]\textsc{p-Median}\hfill\\
        \textbf{Instances:} Set of potential facilities $F = \fromto{1}{n}$, set of clients $C = \fromto{1}{m}$, service cost function $c: F \times C \rightarrow \Z$, facility threshold $p \in \N$, cost threshold $k \in \Z$\\
        \textbf{Universe:} Facility set $F =: \U$.\\
        \textbf{Solution set:} The set of sets $F' \subseteq F$ s.t. $|F'| \leq p$ and $\sum_{j \in C} \min_{i \in F'} c(i, j) \leq k$.
    	\end{description}
    \end{mdframed}
\end{samepage}

\begin{samepage}
    \begin{mdframed}
    	\begin{description} 
        \item[]\textsc{Independent Set}\hfill\\
        \textbf{Instances:} Graph $G = (V,E)$, number $k \in \N$.\\
        \textbf{Universe:} Vertex set $V =: \U$.\\
        \textbf{Feasible solution set:} The set of all independent sets.\\
        \textbf{Solution set:} The set of all independent sets of size at least $k$.
    	\end{description}
    \end{mdframed}
\end{samepage}

\begin{samepage}
    \begin{mdframed}
    	\begin{description}   
        \item[]\textsc{Clique}\hfill\\
        \textbf{Instances:} Graph $G = (V, E)$, number $k \in \N$.\\
        \textbf{Universe:} Vertex set $V =: \U$.\\
        \textbf{Feasible solution set:} The set of all cliques.\\
        \textbf{Solution set:} The set of all cliques of size at least $k$.
    	\end{description}
    \end{mdframed}
\end{samepage}

\begin{samepage}
    \begin{mdframed}
    	\begin{description}   
        \item[]\textsc{Subset Sum}\hfill\\
        \textbf{Instances:} Numbers $\fromto{a_1}{a_n} \subseteq \N$, and target value $M \in \N$.\\
        \textbf{Universe:} $\fromto{a_1}{a_n} =: \U$.\\
        \textbf{Solution set:} The set of all sets $S \subseteq \U$ with $\sum_{a_i \in S}a_i = M$.
    	\end{description}
    \end{mdframed}
\end{samepage}

\begin{samepage}
    \begin{mdframed}
    	\begin{description}   
        \item[]\textsc{Knapsack}\hfill\\
        \textbf{Instances:} Objects with prices and weights $\fromto{(p_1, w_1)}{(p_n, w_n)} \subseteq \N^2$, and $W, P \in \N$.\\
        \textbf{Universe:} $\fromto{(p_1, w_1)}{(p_n, w_n)} =: \U$.\\
        \textbf{Feasible solution set:} The set of all $S \subseteq \U$ with $\sum_{(p_i, w_i) \in S}w_i \leq W$.\\
        \textbf{Solution set:} The set of feasible $S$ with $\sum_{(p_i, w_i) \in S} p_i \geq P$.
    	\end{description}
    \end{mdframed}
\end{samepage}

\begin{samepage}
    \begin{mdframed}
    	\begin{description}
        \item[]\textsc{Directed Hamiltonian Path}\hfill\\
        \textbf{Instances:} Directed Graph $G = (V, A)$, Vertices $s, t \in V$.\\
        \textbf{Universe:} Arc set $A =: \U$.\\
        \textbf{Solution set:} The set of all sets $C \subseteq A$ forming a Hamiltonian path going from $s$ to $t$.
    	\end{description}
    \end{mdframed}
\end{samepage}

\begin{samepage}
    \begin{mdframed}
    	\begin{description}
        \item[]\textsc{Directed Hamiltonian Cycle}\hfill\\
        \textbf{Instances:} Directed Graph $G = (V, A)$.\\
        \textbf{Universe:} Arc set $A =: \U$.\\
        \textbf{Solution set:} The set of all sets $C \subseteq A$ forming a Hamiltonian cycle.
    	\end{description}
    \end{mdframed}
\end{samepage}

\begin{samepage}
    \begin{mdframed}
    	\begin{description}  
        \item[]\textsc{Undirected Hamiltonian Cycle}\hfill\\
        \textbf{Instances:} Graph $G = (V, E)$.\\
        \textbf{Universe:} Edge set $E =: \U$.\\
        \textbf{Solution set:} The set of all sets $C \subseteq E$ forming a Hamiltonian cycle.
    	\end{description}
    \end{mdframed}
\end{samepage}

\begin{samepage}
    \begin{mdframed}
    	\begin{description}
        \item[]\textsc{Traveling Salesman Problem}\hfill\\
        \textbf{Instances:} Complete Graph $G = (V, E)$, weight function $w: E \rightarrow \Z $, number $k \in \N$.\\
        \textbf{Universe:} Edge set $E =: \U$.\\
        \textbf{Feasible solution set:} The set of all TSP tours $T\subseteq E$.\\
        \textbf{Solution set:} The set of feasible $T$ with $w(T) \leq k$.
    	\end{description}
    \end{mdframed}
\end{samepage}

\begin{samepage}
    \begin{mdframed}
    	\begin{description}
        \item[]\textsc{Directed} $k$-\textsc{Vertex Disjoint Path}\hfill\\
        \textbf{Instances:} Directed graph $G = (V, A)$, $s_i, t_i \in V$ for $i \in \fromto{1}{k}$.\\
        \textbf{Universe:} Arc set $A =: \U$.\\
        \textbf{Solution set:} The sets of all sets $A' \subseteq A$ such that $A' = \bigcup^k_{i = 1} A(P_i)$, where all $P_i$ are pairwise vertex-disjoint paths from $s_i$ to $t_i$ for $1 \leq i \leq k$.
    	\end{description}
    \end{mdframed}
\end{samepage}

\begin{samepage}
    \begin{mdframed}
    	\begin{description}
        \item[]\textsc{Steiner Tree}\hfill\\
        \textbf{Instances:} Undirected graph $G = (S \cup T, E)$, set of Steiner vertices $S$, set of terminal vertices $T$, edge weights $c: E \rightarrow \N$, number $k \in \N$.\\
        \textbf{Universe:} Edge set $E =: \U$.\\
        \textbf{Feasible solution set:} The set of all sets $E' \subseteq E$ such that $E'$ is a tree connecting all terminal vertices from $T$.\\
        \textbf{Solution set:} The set of feasible solutions $E'$ with $\sum_{e' \in E'} c(e') \leq k$.
    	\end{description}
    \end{mdframed}
\end{samepage}

\end{document}